\documentclass[11pt]{amsart}
\usepackage[margin=1in,letterpaper]{geometry}                
\usepackage{amsmath,xcolor}
\usepackage{amssymb,wasysym}
\usepackage[colorlinks,
             linkcolor=blue,
             citecolor=green!50!black,
             pdfproducer={pdfLaTeX},
             pdfpagemode=None,
             bookmarksopen=true
             bookmarksnumbered=true]{hyperref}

\newcommand\arXiv[1]{\href{http://arxiv.org/abs/#1}{\nolinkurl{arXiv:#1}}}
\newcommand\MRnumber[1]{\href{http://www.ams.org/mathscinet-getitem?mr=#1}{\nolinkurl{MR#1}}}
\newcommand\DOI[1]{\href{http://dx.doi.org/#1}{\nolinkurl{DOI:#1}}}
\newcommand\MAILTO[1]{\href{mailto:#1}{\nolinkurl{#1}}}

\usepackage{amsmath,amsthm,cleveref}

\newtheorem{theorem}[subsection]{Theorem}
\newtheorem{conjecture}[subsection]{Conjecture}
\newtheorem{lemma}[subsection]{Lemma}
\newtheorem{proposition}[subsection]{Proposition}
\newtheorem{corollary}[subsection]{Corollary}

\theoremstyle{definition} 

\newtheorem{definitionnodiamond}[subsection]{Definition}
\newtheorem{examplenodiamond}[subsection]{Example}
\newtheorem{remarknodiamond}[subsection]{Remark}
\newtheorem{warningnodiamond}[subsection]{Warning}

\newenvironment{definition}{\begin{definitionnodiamond}}{\hfill\ensuremath\Diamond\end{definitionnodiamond}}
\newenvironment{example}{\begin{examplenodiamond}}{\hfill\ensuremath\Diamond\end{examplenodiamond}}
\newenvironment{remark}{\begin{remarknodiamond}}{\hfill\ensuremath\Diamond\end{remarknodiamond}}

\renewcommand\qedhere{\qed}
\newcommand\diamondhere{\hfill\ensuremath\Diamond}

\numberwithin{equation}{subsection}

\crefname{section}{Section}{Section}
\crefformat{section}{#2Section~#1#3} 
\Crefformat{section}{#2Section~#1#3} 

\crefname{conjecture}{Conjecture}{Conjectures}
\crefformat{conjecture}{#2Conjecture~#1#3} 
\Crefformat{conjecture}{#2Conjecture~#1#3} 

\crefname{definition}{Definition}{Definitions}
\crefformat{definition}{#2Definition~#1#3} 
\Crefformat{definition}{#2Definition~#1#3} 

\crefname{definitionnodiamond}{Definition}{Definitions}
\crefformat{definitionnodiamond}{#2Definition~#1#3} 
\Crefformat{definitionnodiamond}{#2Definition~#1#3} 

\crefname{example}{Example}{Examples}
\crefformat{example}{#2Example~#1#3} 
\Crefformat{example}{#2Example~#1#3} 

\crefname{examplenodiamond}{Example}{Examples}
\crefformat{examplenodiamond}{#2Example~#1#3} 
\Crefformat{examplenodiamond}{#2Example~#1#3} 

\crefname{remark}{Remark}{Remarks}
\crefformat{remark}{#2Remark~#1#3} 
\Crefformat{remark}{#2Remark~#1#3} 

\crefname{remarknodiamond}{Remark}{Remarks}
\crefformat{remarknodiamond}{#2Remark~#1#3} 
\Crefformat{remarknodiamond}{#2Remark~#1#3} 

\crefname{warning}{Warning}{Warnings}
\crefformat{warning}{#2Warning~#1#3} 
\Crefformat{warning}{#2Warning~#1#3} 

\crefname{warningnodiamond}{Warning}{Warnings}
\crefformat{warningnodiamond}{#2Warning~#1#3} 
\Crefformat{warningnodiamond}{#2Warning~#1#3}

\crefname{lemma}{Lemma}{Lemmas}
\crefformat{lemma}{#2Lemma~#1#3} 
\Crefformat{lemma}{#2Lemma~#1#3} 

\crefname{proposition}{Proposition}{Propositions}
\crefformat{proposition}{#2Proposition~#1#3} 
\Crefformat{proposition}{#2Proposition~#1#3} 

\crefname{corollary}{Corollary}{Corollaries}
\crefformat{corollary}{#2Corollary~#1#3} 
\Crefformat{corollary}{#2Corollary~#1#3} 

\crefname{theorem}{Theorem}{Theorems}
\crefformat{theorem}{#2Theorem~#1#3} 
\Crefformat{theorem}{#2Theorem~#1#3}

\crefname{equation}{}{}
\crefformat{equation}{(#2#1#3)} 
\Crefformat{equation}{(#2#1#3)}

\usepackage{tikz}
\usetikzlibrary{decorations.pathreplacing,decorations.markings,arrows,shapes}
\tikzset{
    dot/.style={circle,draw,fill,inner sep=1pt},
    arrow/.style={->,thick,shorten <=2pt,shorten >=2pt},
    twoarrow/.style={double,double distance=1.5pt,shorten <=9pt,shorten >=10pt,decoration={markings,mark=at position -8pt with {\arrow[scale=2]{>}}},preaction={decorate}},
    twoarrowlonger/.style={double,double distance=1.5pt,shorten <=5pt,shorten >=6pt,decoration={markings,mark=at position -4pt with {\arrow[scale=2]{>}}},preaction={decorate}},
    twoarrowshorter/.style={double,double distance=1.5pt,shorten <=13pt,shorten >=14pt,decoration={markings,mark=at position -12pt with {\arrow[scale=2]{>}}},preaction={decorate}},
    twoarrowshorthead/.style={double,double distance=1.5pt,shorten <=9pt,shorten >=20pt,decoration={markings,mark=at position -18pt with {\arrow[scale=2]{>}}},preaction={decorate}},
    threearrowpart1/.style={ thick,double,double distance=3pt,shorten <=9pt,shorten >=11pt},
    threearrowpart2/.style={ thick,shorten <=9pt,shorten >=10pt},
    threearrowpart3/.style={ shorten <=9pt,shorten >=10pt,decoration={markings,mark=at position -8pt with {\arrow[scale=3]{>}}},preaction={decorate}},
}

\newcommand\define[1]{\emph{#1}}

\newcommand\cat[1]{\textsc{#1}}

\newcommand\bC{\mathbb C}

\newcommand\bE{\mathbb E}

\newcommand\bK{\mathbb K}

\newcommand\bN{\mathbb N}

\newcommand\bP{\mathbb P}

\newcommand\bR{\mathbb R}

\newcommand\bZ{\mathbb Z}

\newcommand\cA{\mathcal A}
\newcommand\cB{\mathcal B}
\newcommand\cC{\mathcal C}
\newcommand\cD{\mathcal D}
\newcommand\cE{\mathcal E}
\newcommand\cF{\mathcal F}
\newcommand\cG{\mathcal G}
\newcommand\cH{\mathcal H}

\newcommand\cM{\mathcal M}

\newcommand\cO{\mathcal O}
\newcommand\cP{\mathcal P}

\newcommand\cS{\mathcal S}

\newcommand\cX{\mathcal X}
\newcommand\cY{\mathcal Y}
\newcommand\cZ{\mathcal Z}

\usepackage{dsfont}
\newcommand\unit{\mathds{1}}

\usepackage[utf8]{inputenc}
\DeclareFontFamily{U}{min}{}
\DeclareFontShape{U}{min}{m}{n}{<-> udmj30}{}

\renewcommand\qedhere{\hfill\ensuremath\Box}

\DeclareMathOperator\End{End}

\DeclareMathOperator\Ann{Ann}
\newcommand\id{\mathrm{id}}
\DeclareMathOperator\EW{EW}
\newcommand\op{\mathrm{op}}

\newcommand\st{\text{ s.t.\ }}

\newcommand\sminus\smallsetminus
\newcommand\from\leftarrow
\newcommand\isom{\overset\sim\to}
\newcommand\mono{\hookrightarrow}

\newcommand\pt{\{\mathrm{pt}\}}

\title{Heisenberg-picture quantum field theory}
\author{Theo Johnson-Freyd}
\email{\MAILTO{theojf@pitp.ca}}
\address{Perimeter Institute for Theoretical Physics, Waterloo, ON, Canada}
\thanks{I would like to thank Alexandru Chirvasitu, Owen Gwilliam, and Claudia Scheimbauer for many ongoing discussions about these and related ideas, to the referee for their helpful comments, and to Stephan Stolz and Peter Teichner for sharing  the construction developed in \cref{factorization algebra example}. Most importantly, my thanks go to Kolya Reshetikhin, whose work on quantum field theory and ribbon categories (e.g.\ \cite{MR1036112,MR1091619,KolyaLong,MR2648332}) inspired this paper.
This work is supported by the grant DMS-1304054.}

\begin{document}

\maketitle

\begin{center}
\emph{For Kolya Reshetikhin on the occasion of his 60th birthday.}
\end{center}

\begin{abstract}
What we should mean by ``Heisenberg-picture quantum field theory''?  Atiyah--Segal-type axioms do a good job of capturing the ``Schr\"odinger picture'': these axioms define a ``$d$-dimensional quantum field theory'' to be a symmetric monoidal functor from an $(\infty,d)$-category of ``spacetimes'' to an $(\infty,d)$-category which at the second-from-top level consists of vector spaces, so at the top level consists of numbers.  This paper argues that the appropriate parallel notion ``Heisenberg picture'' should also be defined in terms of symmetric monoidal functors from the category of spacetimes, but the target should be an $(\infty,d)$-category that in top dimension consists of pointed vector spaces instead of numbers; the second-from-top level can be taken to consist of associative algebras or of pointed categories.
The paper ends by outlining two sources of such Heisenberg-picture field theories: factorization algebras and skein theory.
\end{abstract}

\section{Introduction and motivation}

Open the nearest book called \emph{Introduction to Quantum Mechanics} --- you probably have one lying around somewhere.
Almost certainly somewhere in it is a discussion of the two so-called ``pictures'' of quantum mechanics, named after Erwin Schr\"odinger and Werner Heisenberg.  The ``Schr\"odinger picture'' has a natural generalization to quantum field theory via Atiyah--Segal-type axioms.  My goal in  this paper is to motivate, propose, and expound upon axioms for a ``Heisenberg picture'' of quantum field theory.  
Let us, then, begin by learning what we can from that \emph{Quantum Mechanics} textbook.
 The discussion of the Schr\"odinger picture translates well into a mathematical definition:

\begin{definitionnodiamond}\label{defn.SchrodingerQM}
  A \define{Schr\"odinger-picture quantum mechanical system} consists of the following data:
  \begin{enumerate}
    \item A vector space $V$ (over some ground field $\bK$) called \define{the space of states}.
    \item For each ``time'' $t\in \bR_{>0}$, a linear map $u_{t} : V \to V$ called \define{time evolution}.  These should satisfy a \define{group law} $u_{t_{1} + t_{2}} = u_{t_{1}}\circ u_{t_{2}}$.
    \item Some distinguished \define{states} $v_{i} \in V$ depending on some labeling set $I$ which our laboratory friends know how to prepare, and some distinguished \define{costates} $w_{j} : V \to \bK$ that we know how to ``post-pare'' (or is it ``post-pair''?).  There may also be some distinguished \define{observables} $a_{k} \in \End(V)$ that we know how to test by making some manipulation while the experiment runs. \diamondhere
  \end{enumerate}
\end{definitionnodiamond}

\begin{remark}
The axioms for a group (a set with a binary operation satisfying~\dots) don't communicate what questions one might want to know about a group.  Moreover, the axioms for a group are often too liberal: groups ``in nature'' are usually Lie or algebraic or finite or hyperbolic or otherwise more structured than the axioms would suggest.  Similarly, Schr\"odinger-picture quantum mechanical systems are usually defined over $\bC$ specifically, the spaces of states are usually Hilbert spaces, time evolution operators are usually unitary,  the sets of distinguished states and costates usually agree, and observables are usually self-adjoint.  

But none of these extra conditions are what that textbook of yours claims is the essence of quantum mechanics, which is the linearity illustrated by the two-slit experiment.  So I won't put it in the definition.  Indeed, Hilbert structures and unitarity are closely related to time-reversal symmetry (recent discussions are available in \cite{FreedHopkins,MR3623677}); our axiomatics should accommodate non-symmetric examples as well.
General questions one might want to ask about a Schr\"odinger-picture quantum mechanical system include the spectrum of  time evolution and the (perhaps asymptotic) values of various compositions of the data.
\end{remark}

\begin{remark}\label{remark.projective}
The discussion of the Schr\"odinger picture in that textbook of yours probably also includes an important comment emphasizing that the space $V$ itself isn't ``physical'': only its projectivization $\bP V = (V \sminus \{0\}) / \bK^{\times}$ is.  It is often tempting to ignore this comment.
\end{remark}

The discussion of the Heisenberg picture begins by emphasizing that in physics, what's most physical are the \define{observables} in a given system, and that in quantum mechanics these form an associative but generally noncommutative algebra.  (Just like how physically-interesting Schr\"odinger-picture quantum mechanics deals with Hilbert spaces rather than general vector spaces, physically-interesting Heisenberg-picture quantum mechanics deals with algebras equipped with extra structure --- C-star or von Neumann, for example --- but we should not build such structure into the basic axiomatics.)  Let us recall the main examples and then try to extract a definition:

\begin{example}\label{eg.heis-quantum}
  Let $(V, u_{t},\dots)$ be a Schr\"odinger-picture quantum mechanical system such that the time evolution operators $u_{t}$ are all isomorphisms.  The corresponding \define{Heisenberg-picture quantum mechanical system} consists of the following data:
  \begin{enumerate}
    \item The associative algebra $A = \End(V)$ of all endomorphism of $V$ is the \define{algebra of observables} of the system.
    \item For each $t\in \bR_{>0}$,  \define{evolution for time $t$} is encoded by the ``conjugate by $u_{t}$'' algebra automorphism $a \mapsto u_{t}au_{t}^{-1}$.
    \item The distinguished observables $a_{k} \in A$ already do not reference $V$.  
    To remove $V$ from the data of the distinguished states and costates, we can encode them  as ideals: given $v_{i} \in V$, remember the left ideal $\Ann(v_{i}) = \{a\in A \st av_{i} = 0\}$; for $w_{j}: V\to\bK$, use the right ideal $\Ann(w_{j}) = \{a\in A \st w_{j}\circ a = 0\}$. 
  \end{enumerate}
  This construction is discussed in~\cite{MR2530165}, which does not answer the question of what to do when the operators $u_t$ are not isomorphisms.  I will propose an answer in \cref{prop.end}.
\end{example}

\begin{remark} \label{remark.direction conventions}
  Some contravariance in the constructions later in this paper is unavoidable.  Just to fix a convention, let's declare (as is most standard) that $\End(V)$ acts on $V$ from the left.  For a general category $\cC$ and object $C\in \cC$ therein, the algebra $\End_\cC(C)$ of endomorphisms of $C$ has multiplication $fg = f\circ g = (C \overset g \to C \overset f \to C)$.  My hope is that the reader will largely ignore left/right questions.
\end{remark}

This example illustrates already one feature of the Heisenberg picture: it solves the problem from \cref{remark.projective} that only $\bP V$ is physical, not $V$.  Indeed, multiplying $u_{t}$ by a non-zero constant does not change the conjugation map $a \mapsto u_{t}au_{t}^{-1}$, and similarly the ideals $\Ann(v_{i})$ and $\Ann(w_{j})$  depend only on the lines spanned by $v_{i}$ and $w_{j}$.  Moreover, in many situations a vector or Hilbert space $V$ can be recovered up to non-unique isomorphism from the algebra $\End(V)$; the ambiguity in recovering $V$ is exactly the group of invertible scalars, so that $\End(V)$ and $\bP V$ encode exactly the same information.

The following example illustrates the other main feature of the Heisenberg picture: it is flexible enough to accommodate classical as well as quantum mechanical systems (and systems intermediate between these extremes, corresponding to non-commutative algebras that are not as fully-noncommutative as $\End(V)$):

\begin{examplenodiamond}\label{eg.heis-classical}
  A \define{classical mechanical system} consists of a symplectic manifold $X$ and a family of symplectomorphisms $\{\varphi_{t} : X \to X\}_{t\in \bR_{>0}}$ satisfying a group law, and perhaps some other distinguished data.  Such a system can be encoded in the Heisenberg picture as follows:
  \begin{enumerate}
    \item The \define{algebra of observables} is the commutative algebra $\cO(X)$ of smooth functions on $\cM$.
    \item \define{Time evolution} is encoded by the algebra isomorphisms $f_{t}=\varphi_{t}^{*} : \cO(X) \to \cO(X)$.
    \item One type of distinguished data might be submanifolds of ``boundary conditions.''  These can be encoded by ideals in   $\cO(X)$. \diamondhere
  \end{enumerate}
\end{examplenodiamond}

\Cref{eg.heis-quantum,eg.heis-classical} suggest the following definition of the Heisenberg picture:

\begin{definitionnodiamond}[Tentative] \label{defn.heisenberg-tentative}
  A \define{Heisenberg-picture quantum mechanical system} consists of:
  \begin{enumerate}
    \item An associative \define{algebra of observables} $A$.
    \item For each time $t \in \bR_{>0}$, a \define{time evolution} homomorphism $f_{t} : A \to A$.  These should satisfy a \define{group law} $f_{t_{1}+t_{2}} = f_{t_{1}}\circ f_{t_{2}}$.  
    \item Some distinguished right ideals encoding initial states/boundary conditions, some distinguished left ideals encoding terminal costates/boundary conditions, and some distinguished elements of $A$ encoding available experimental manipulations.\diamondhere
  \end{enumerate}
\end{definitionnodiamond}

\cref{defn.heisenberg-tentative} is only tentative, and will be revised in the next section.  For instance, it is not clear how (or whether?)\ to accommodate Schr\"odinger-picture systems with non-invertible time evolution.

\section{Modulation}

The Schr\"odinger picture of \cref{defn.SchrodingerQM} generalizes naturally to an Atiyah--Segal style axiomatics for quantum field theory.  Indeed, 
\cref{defn.SchrodingerQM}, and in particular the group law, already defines quantum mechanical systems as a type of functor:

\begin{definition}\label{defn.qmspacetimes}
  Let $I,J,K$ be label sets.
  The category $\cat{QMSpacetimes}$ of \define{quantum mechanical spacetimes} with ``point defect'' label sets $I,J,K$ has:
  \begin{description}
    \item[Objects] Two objects, called $\pt$ and $\emptyset$.
    \item[Generating morphisms] For each $t\in \bR_{>0}$, there is a morphism $u_{t} : \pt \to \pt$, which can be thought of as a metrized interval of length $t$.  These are required to satisfy the group law relation $u_{t_{1}}\circ u_{t_{2}} = u_{t_{1}+t_{2}}$.
    
    There are also morphisms $v_{i} : \emptyset \to \pt$ for each $i\in I$, $w_{j} : \pt \to \emptyset$ for each $j\in J$, and $a_{k} : \pt \to \pt$, called \define{point defects}.  We impose no relations on these.
  \end{description}
  The full category then consists of compositions of the generating morphisms (modulo the group law relation).
\end{definition}

Let $\cat{Vect}$ denote the category of $\bK$-vector spaces and linear maps.  \cref{defn.SchrodingerQM} can then be rephrased to state that a Schr\"odinger-picture quantum mechanical system is a functor $V : \cat{QMSpacetimes} \to \cat{Vect}$ along with an isomorphism $V(\emptyset) \cong \bK$.

\begin{remark}\label{remark.strongpointedfunctor}
  $\cat{Vect}$ and $\cat{QMSpacetimes}$ are examples of \define{pointed categories}: categories equipped with distinguished objects (namely $\bK \in \cat{Vect}$ and $\emptyset \in \cat{QMSpacetimes}$).  Then $V$ is precisely a \define{strong pointed functor} from $\cat{Vect}$ to $\cat{QMSpacetimes}$ (compare \cref{lax pointed map}).  
  
  One can alternately freely generate a symmetric monoidal category from $\cat{QMSpacetimes}$ under the condition that $\emptyset$ becomes the monoidal unit.  Calling the monoidal structure $\sqcup$, the objects of this freely-generated category are finite sets --- formal disjoint unions of copies of $\pt$ --- and its morphisms are disjoint unions of directed metrized intervals with ``point defects'' labeled by elements of $I$, $J$, and $K$.
\end{remark}

This functorial point of view can be applied to higher-dimensional spacetimes:

\begin{definition}[\cite{MR1001453,MR2079383}]
  Given a geometry $\cG$, write $\cat{Spacetimes}_{d-1,d}^\cG$ for the symmetric monoidal category whose objects are 
  $(d-1)$-dimensional manifolds equipped with geometry of type $\cG$ and whose morphisms are isomorphism types of $d$-dimensional cobordisms equipped with geometry of type $G$.     
  A \define{Schr\"odinger-picture ($d$-dimensional, non-extended, for geometry $\cG$) quantum field theory} is a symmetric monoidal functor $\cat{Spacetimes}_{d-1,d}^\cG \to \cat{Vect}$.
\end{definition}

\begin{remark}
The details of, and potential difficulty in constructing, the category $\cat{Spacetimes}_{d-1,d}^\cG$ depends on the choice of geometry $\cG$.
The notion of ``geometry'' should be interpreted quite liberally: it could include metrics, background fields, labeled defects, etc.  For some discussion and examples of geometric cobordism categories, see \cite{AyalaThesis,MR2648332,MR2742432}.  When $\cG$ is sufficiently topological, $\cat{Spacetimes}_{d-1,d}^\cG$ can be built from the category $\cat{Bord}_{d-1,d}^{\mathrm{smooth}}$ of smooth but otherwise unstructured spacetimes as in \cite[Section 3.2]{Lur09}, elaborated upon in \cite{MR3623677}.
\end{remark}

Ignoring for the moment the boundary-condition ideals and distinguished manipulations (item~(3) in \cref{defn.heisenberg-tentative}), a Heisenberg-picture quantum mechanical system is also a functor: its source is the category $\cat{QMSpacetimes}$ of metric segments from \cref{defn.qmspacetimes}, and its target is the category $\cat{Alg}^{\mathrm{homo}}$ whose objects are associative algebras and whose morphisms are algebra homomorphisms (we will soon introduce a different category called $\cat{Alg}$, hence the name for this one).  Let $\cat{Vect}^{\mathrm{inv}}$ denote the maximal subgroupoid of $\cat{Vect}$: it has the same objects, but morphisms must be isomorphisms.  There is a functor $\End : \cat{Vect}^{\mathrm{inv}} \to \cat{Alg}^{\mathrm{homo}}$ which sends a vector space $V$ to its endomorphism algebra $\End(V)$ and an isomorphism $f : V \to W$ to  the algebra homomorphism $a \mapsto faf^{-1} : \End(V) \to \End(W)$.  
\cref{eg.heis-quantum} then consists of taking in a Schr\"odinger-picture quantum mechanical system $V : \cat{QMSpacetimes} \to \cat{Vect}^{\mathrm{inv}}$ and producing the composition $\cA = {\End} \circ V : \cat{QMSpacetimes} \to \cat{Alg}^{\mathrm{homo}}$.
This functor is used in \cite{MR2530165} to turn Schr\"odinger-picture quantum field theories valued in $\cat{Vect}^{\mathrm{inv}}$ into nets of algebras in the sense of ``axiomatic'' or ``algebraic'' quantum field theory.

The question then arises: can we similarly compose arbitrary Schr\"odinger-picture quantum field theories with this $\End$ functor to produce Heisenberg-picture quantum field theories?  
My goal, of course, is to explain that the answer is ``yes.''  But the main obstruction, hinted at already in \cref{eg.heis-quantum} and after \cref{defn.heisenberg-tentative}, must be confronted: in most Schr\"odinger-picture quantum field theories, the linear maps associated to spacetimes are not invertible.
To resolve this obstruction, let us seek guidance from another example:

\begin{example}\label{eg.classicalfieldtheory}
  Typical classical field theories correspond to partial differential equations.  Let $D$ be some partial differential equation for fields on $d$-dimensional $\cG$-geometric manifolds.  Then the classical field theory $\cF$ assigns to a $(d-1)$-dimensional manifold $N$ the space $\cF(N) = \{$germs of solutions to the PDE $D$ on a small $d$-dimensional neighborhood of $N\}$ and to a $d$-dimensional cobordism $M$ the space $\cF(M) = \{$solutions to $D$ on $M\}$.  These data package into a \define{span of spaces}: to any cobordism $(M : N_{1} \to N_{2}) \in \cat{Spacetimes}^\cG_{d-1,d}$, we get the span $\cF(N_{1}) \from \cF(M) \to \cF(N_{2})$.
  
  The corresponding Heisenberg-picture TQFT should assign to $N$ the algebra $\cO(\cF(N))$ of functions on $\cF(N)$.  If the span $\cF(N_{1}) \from \cF(M) \to \cF(N_{2})$ were the graph of a function $f_{M} : \cF(N_{1}) \from \cF(N_{2})$, then we could associative in the Heisenberg picture the homomorphism $f_{M}^{*} : \cO(\cF(N_{1})) \to \cO(\cF(N_{2}))$.  But typically it is not the graph of a function.  Applying $\cO$ to all pieces produces a cospan of algebras $\cO(\cF(N_{1})) \to \cO(\cF(M)) \from \cO(\cF(N_{2}))$.  
  This is about as far as abstract nonsense can take us.
\end{example}

Perhaps we should follow \cref{eg.classicalfieldtheory} and expect that Heisenberg-picture TQFTs should assign to cobordisms cospans of associative algebras?  This is a tempting answer, but has a few problems.  It doesn't fully accommodate the ``boundary ideals'' from item~(3) of \cref{defn.heisenberg-tentative} --- if these were two-sided ideals, we could instead use the quotient rings, but generically they are only one-sided.  Moreover, the canonical way to compose cospans is via a push-out square.  These exist in the category of associative algebras, but are too large (about as large as a free associative algebra).  Considering further the classical field theory of \cref{eg.classicalfieldtheory}, one notes that the composition of spans of spaces, given by a pullback of spaces, corresponds to the pushout of commutative algebras, which is also the tensor product.  Given cospans of associative algebras $A_{0} \to B_{1} \from A_{1} \to B_{2} \from A_{2}$, one can give $B_{1}$ a right $A_{1}$-module structure and $B_{2}$ a left $A_{1}$-module structure, and consider the composition $B_{1} \otimes_{A_{1}} B_{2}$.  But this generically is not an associative algebra.

On the other hand, it is a module, as are ideals.  And cospans of commutative algebras are examples of bimodules.  Perhaps (bi)modules are the appropriate target of Heisenberg-picture field theory?  This will not quite be our final definition, but it is sufficiently important as to merit its own name:

\begin{definition}
  The \define{Morita} bicategory $\cat{Mor} = \cat{Mor}_{1}(\cat{Vect}_{\bK})$ over the field $\bK$ has:
  \begin{description}
    \item[Objects] Associative algebras over $\bK$.
    \item[1-morphisms] A 1-morphism between associative algebras $A$ and $B$ is an $A$-$B$-bimodule $_{A} M _{B}$.  These compose by tensor product.
    \item[2-morphisms] A 2-morphism is a homomorphism of bimodules.
  \end{description}
  This bicategory is symmetric monoidal with the usual tensor product over $\bK$ \cite{MR2534210,Shulman2010}.
  
  A \define{(non-extended) Morita-picture quantum field theory} for the spacetime category $\cat{Spacetimes}$ is a symmetric monoidal functor $\cat{Spacetimes} \to \cat{Mor}$.
\end{definition}

How can we translate a field theory in the sense of \cref{defn.heisenberg-tentative} into a Morita-picture field theory?
It's not yet clear how to accommodate the distinguished elements of $\cA$ from item~(3), but the distinguished ideals are already modules, hence easy to handle.  As for time evolution, one can always turn homomorphisms into modules:

\begin{definitionnodiamond}\label{defn.modulation}
  Let $f : A \to B$ be a homomorphism of associative algebras.  The \define{modulation} of $f$ is the $A$-$B$-bimodule $\cM(f) = {_{f} B}$, which as a left $B$-module is just $B$, and has a right $A$-action via~$f$: \\
  \mbox{} \hfill $a \triangleright b' \triangleleft b = f(a)\,b'\,b$. \diamondhere
\end{definitionnodiamond}

One may easily check that modulation defines a functor $\cM : \cat{Alg}^{\mathrm{homo}} \to \cat{Mor}$.  The name is from~\cite{MR2304628}.

\section{The point of pointings}\label{section.pointings}

Unfortunately, the modulation functor of \cref{defn.modulation}  loses too much information:

\begin{lemma}\label{lemma.morisn'tbetter}
  Let $f,g : A\to B$ be homomorphisms of associative algebras.   Then $\cM(f) \cong \cM(g) \in \cat{Mor}$ if and only if there exists  an invertible element $b\in B$ such that for each $a\in A$, $f(a) = b^{-1}\,g(a)\,b$.
\end{lemma}
\begin{proof}
  Given such a $b \in B$, the map $\cM(f) \to \cM(g)$ given by multiplication by $b$ on the left is an isomorphism of $A$-$B$-bimodules.  The converse is an easy exercise for the reader.
\end{proof}

In particular, one can fully recover  from its modulation the Heisenberg-picture encoding of a classical mechanical system in the sense of \cref{eg.heis-classical}.  However, if one starts with a Schr\"odinger picture quantum mechanical system, applies \cref{eg.heis-quantum}, and then modulates the output, all information is lost: the functor $\cat{QMSpacetimes} \to \cat{Mor}$ produced in this way does not depend (up to isomorphism) on the time evolution operators $U_{t}$.

To fix this requires breaking the multiplication by $b$ in the proof of \cref{lemma.morisn'tbetter}.  A minimal way to do this is to remember one extra bit of data: which element of $\cM(f)$ corresponds to $1_{B}\in B$.  With this extra information, homomorphisms $f:A \to B$ can be recovered.  Indeed, given the $A$-$B$-bimodule $\cM(f) = {_{f}B}$ and the vector $1_{B} \in {_{f}B}$, the element $f(a) \in B$ for a given $a\in A$ is the unique solution to the equation $a\triangleright 1_{B} = 1_{B}\triangleleft f(a)$.  This suggests that we revise \cref{defn.modulation}: the output of modulation is not just a bimodule, but a \define{pointed} bimodule.  

\begin{definition}\label{defn.alg}
  The bicategory $\cat{Alg} = \cat{Alg}_{1}(\cat{Vect}_{\bK})$ of \define{algebras and pointed bimodules} has:
  \begin{description}
    \item[Objects] An object of $\cat{Alg}$ is an associative algebra over $\bK$.
    \item[1-morphisms] A 1-morphism from $A$ to $B$ is an $A$-$B$-bimodule $_{A} M _{B}$ along with a \define{pointing} $1_M \in M$.  The composition of $(_{A} M_{B},1_M\in M)$ with $(_{B} N _{C},1_N\in N)$ is the $A$-$C$-bimodule $M \otimes_{B} N$ pointed by the class of $1_M\otimes 1_N$.
    \item[2-morphisms] A 2-morphism $(_{A}M_{B},1_M\in M) \to (_{A}N_{B},1_N\in N)$ is a bimodule homomorphism $f : M \to N$ such that $f(1_M) = 1_N$.
  \end{description}
  Like $\cat{Mor}$, this bicategory is symmetric monoidal with the usual tensor product over $\bK$.
  
  There is an obvious forgetful functor $\cat{Alg} \to \cat{Mor}$ which forgets all pointings.  The modulation functor factors through it: abusing notation, we let $\cM : \cat{Alg}^{\mathrm{homo}} \to \cat{Alg}$ denote the {modulation} functor that sends a homomorphism $f: A \to B$ to the pointed bimodule $(_{f}B,1_{B})$.
\end{definition}

Consider modulating the Heisenberg-picture quantum mechanical system from \cref{eg.heis-quantum} corresponding to a Schr\"odinger-picture system in which time evolution $u_{t}$ is invertible.  By \cref{lemma.morisn'tbetter}, the modulation of $a \mapsto u_{t} a u_{t}^{-1}$ is isomorphic in $\cat{Mor}$ to the modulation of the identity $\cM(\id_{A}) = {_{A}A_{A}}$, where $A = \End(V)$.  They are not isomorphic in $\cat{Alg}$ when $u_{t}\neq 1$.  The isomorphism in $\cat{Mor}$ consists of multiplication by $u_{t}$.  It follows that:
\begin{lemma}\label{lemma.modconj}
  Given $V \in \cat{Vect}$ and $u : V \isom V$ an isomorphism, let $A = \End(V)$.  The modulation $\cM(a \mapsto uau^{{-1}})$ of conjugation by $u$ is isomorphic in $\cat{Alg}$ to the trivial bimodule $_{A}A_{A}$ pointed not by $1_{A}$ but by the element $u \in A$.\qedhere
\end{lemma}

This suggests that even when time evolution is not an isomorphism, we can nevertheless encode it as a pointed bimodule: the identity bimodule, pointed by time evolution.  Indeed:

\begin{proposition}\label{prop.end}
  There is a contravariant functor $\End: \cat{Vect} \to \cat{Alg}$ taking a vector space $V$ to its endomorphism algebra $\End(V)$ and taking a linear map $f \in \hom(V,W)$ to the pointed bimodule $\bigl(\hom(V,W),f\bigr)$, where $\End(V)$ and $\End(W)$ act on $\hom(V,W)$ by pre- and post-composition.
  
  This functor is symmetric monoidal when restricted to finite-dimensional vector spaces or when $\cat{Vect}$ (and, correspondingly, $\cat{Alg}$) is replaced by an appropriate category of topological vector spaces.  For example, $\End$ is symmetric monoidal if $\cat{Vect}$ is replaced by the category of Hilbert spaces and bounded operators, $\End(V)$ is topologized as a von Neumann algebra in the usual way, and $\cat{Alg}$ consists of von Neumann algebras and pointed Hilbert-space bimodules with the von Neumann tensor product. \qedhere
\end{proposition}

\begin{remark}
  The functor $\End$ is contravariant because, following \cref{remark.direction conventions}, $\hom(V,W)$ carries a \emph{left} action by $\End(W)$ and a \emph{right} action by $\End(V)$.
\end{remark}

We leave checking details to the reader.  \cref{lemma.modconj} assures that the functor $\End$ extends to all of $\cat{Vect}$ the composition ``conjugate, then modulate'' from $\cat{Vect}^{\mathrm{inv}}$ via $\cat{Alg}^{\mathrm{homo}}$ to $\cat{Alg}$.  With \cref{prop.end} in place, we define:

\begin{definition}
A \define{(non-extended, affine) Heisenberg-picture quantum field theory} is a symmetric monoidal functor $\cat{Spacetimes} \to \cat{Alg}$.
\end{definition}

\begin{example}
  Any Schr\"odinger-picture quantum field theory $\cZ : \cat{Spacetimes} \to \cat{Vect}$ determines a Heisenberg-picture quantum field theory $\End\circ \cZ : \cat{Spacetimes} \to \cat{Alg}$ via \cref{prop.end}, provided the values of the functor $\cZ$ are in some subcategory (e.g.\ finite-dimensional vector spaces or Hilbert spaces) for which the functor $\End : \cat{Vect} \to \cat{Alg}$ is symmetric monoidal.
\end{example}

\begin{example}
  Let $\cat{Spans}$ denote the category of spans of spaces.  Then a \define{classical field theory} as in \cref{eg.classicalfieldtheory} is a symmetric monoidal functor $\cF : \cat{Spacetimes} \to \cat{Spans}$.  Any such classical field theory determines a Heisenberg-picture  quantum field theory $\cO\circ \cF : \cat{Spacetimes} \to \cat{Alg}$, where $\cO(X)$ is the algebra of functions on the space $X$.  When $N$ is an object of $\cat{Spacetimes}$ (i.e.\ a $(d-1)$-dimensional manifold with appropriate geometry) we regard $\cO(\cF(N))$ as an associative algebra.  When $M$ is a morphism in $\cat{Spacetimes}$ (i.e.\ a $d$-dimensional manifold with appropriate geometry) we regard $\cO(\cF(M))$ just as a pointed vector space (pointed by its unit element $1_{\cO(\cF(M))}$).
\end{example}

Finally we can update 
 \cref{defn.heisenberg-tentative}:
 \begin{definition}
  A \define{Heisenberg-picture quantum mechanical system} is a Heisenberg-picture quantum field theory for the spacetime category $\cat{QMSpacetimes}$ of one-dimensional spacetimes with point defects from \cref{defn.qmspacetimes} (or, rather, the symmetric monoidal envelope thereof).  
  
  Given a system as in \cref{defn.heisenberg-tentative}, we define the contravariant symmetric monoidal functor $\cH : \cat{QMSpacetimes} \to \cat{Alg}$ by declaring:
  \begin{enumerate}
    \item $\cH(\pt) = A$.  $\cH(\emptyset) = \bK$ by symmetric monoidality.
    \item $\cH(u_{t}) = \cM(f_{t})$ is the modulation of time evolution.
    \item Given a distinguished experimental manipulation $a \in A$, the corresponding point defect is sent to $(_{A}A_{A},a)$, the identity bimodule pointed by $a$.  Given a distinguished left ideal ${_{A}I} \subseteq {_{A}A}$, the corresponding point defect is sent to the left module $A/I$, pointed by the class of $1_{A}$.  Given a distinguished right ideal $J_{A} \subseteq A_{A}$, the corresponding point defect is sent to the right module $J\backslash A$, pointed by the class of $1_{A}$.
  \end{enumerate}
  The main thing to note is the treatment of the distinguished ideals and elements in item~(3): all determine pointed (bi)modules.
  \end{definition}

\begin{remark}
  I have discussed pointed modules as a way of assigning algebras of quantum observables to codimension-one spaces without losing too much information about the algebra.  But pointed modules themselves have a direct interpretation in quantum field theory as the ``purely algebraic part'' of path integrals.  
  
  Indeed, suppose we are given an $n$-dimensional affine variety $X$ of ``field configurations over $M$'' along with a polynomial ``action'' functional $s \in \cO(X)$ and a volume form on $X$.  The Feynman path integral invites us to consider the values of ``oscillating'' integrals $\langle f\rangle = \int_X f\,e^s\,\mathrm{dVol}$ for polynomial ``observables'' $f$.  This integral is insufficiently defined: to define it requires choosing a contour in $X$ along which integrals against the measure $e^s\,\mathrm{dVol}$ converge; up to homotopies of contours that leave all integrals unchanged, the space of contours is parameterized by the relative cohomology group $\mathrm H_n(X;\{\Re(s)\ll 0\})$.
  
  The purely algebraic part of integration is the calculation of the class of $f$ in the quotient $\cO(X) / (\text{total derivatives})$, as the integrals of total derivative vanish on any contour.  This vector space is  naturally pointed by the class of $1$, and is isomorphic (via multiplication by $\mathrm{dVol}$) to the $n$th cohomology group of the twisted de Rham complex for $s$.   In good situations,  the class of $f$ in $\cO(X) / (\text{total derivatives})$ (or of $f\,\mathrm{dVol}$ in the twisted de Rham complex) can be calculated using homotopy algebra \cite{Johnson-Freyd2012}.  In a general Heisenberg-picture quantum field theory $\cZ$, the pointed vector space $\cZ(M)$ assigned to a closed top-dimensional manifold $M$ can be interpreted as ``$\cO(X) /  (\text{total derivatives})$'' for an ill-defined path integral.
\end{remark}

\section{Non-affine field theory}

Do you still have that \emph{Introduction to Quantum Mechanics} textbook?  In one of its more philosophical sections, it is likely to discuss the following basic premise of experimental science: the universe consists only of things that are in principle measurable; if no experiment can distinguish two states, then those states are equal.  The tautologous version of this philosophy is the mathematicians' Yoneda lemma.  But there is a non-tautologous version, which asserts that by ``measurement'' and ``experiment'' we should mean ``element of the algebra of observables.''  For classical phase spaces, for example, the non-tautologous version asserts that points can be separated by real-valued functions --- that all spaces are \define{affine} in the sense of algebraic geometry.  This assertion is true for many types of spaces (smooth manifolds; locally compact Hausdorff spaces) but by no means all spaces: there are many roles in physics for non-affine
schemes and stacks.

A piece of philosophy that probably is not in your textbook is that ``most of 0-algebraic geometry is 1-affine.''  Said another way, although schemes and stacks usually are not determined by their algebras (0-categories) of global functions, they are often determined by their symmetric monoidal (1-)categories of quasicoherent sheaves of modules --- such a category should be understood as ``the algebra of $\cat{Vect}$-valued global functions.''  
(In general, one could call a space  \define{$k$-affine} if it is determined by its symmetric monoidal $k$-category of global maps to the $k$-categorical analogue of $\cat{Vect}$.)
Most algebrogeometric objects one comes across are known to be 1-affine~\cite{Lur09,MR2669705,Brandenburg2011,MR3097055,MR3144607,MR4009673,MR3361309}; \cite{MR4009673} also records a few objects which are known to be non-1-affine.  

It's not my intention to develop here the theory of 1-affine algebraic geometry, but it's worth making a few remarks.  The Gabriel--Rosenberg theorem~\cite{MR0232821,MR1615928} reconstructs a scheme up to isomorphism from its category of quasicoherent modules with no extra structure: only the category itself is used.  Here is a baby case of this result:
let $A$ be a commutative ring and $\cat{Mod}_{A}$ the category of right $A$-modules.  Then $A$ can be reconstructed as the algebra of natural endomorphisms of the identity functor $\id : \cat{Mod}_{A} \to \cat{Mod}_{A}$.

  This has suggested to many workers in ``noncommutative algebraic geometry'' that ``abelian category'' is a good definition of ``noncommutative scheme'' (e.g.~\cite{MR1622759}).  I would argue, however, that this misunderstands the variability of categories.  Consider, for example, the stacks $\pt \sqcup \pt$ and $\pt / (\bZ/2)$, the latter being the classifying stack of the group $\bZ/2$.  Provided we work over a ring in which $2$ is invertible, these have equivalent categories of modules.  But they are honestly different as stacks.  To fully recover $\pt \sqcup \pt$ and $\pt / (\bZ/2)$ from their categories of modules it suffices to remember additionally the symmetric monoidal structures on those categories.  Moreover, the natural homomorphisms of abelian categories are the exact functors, but these do not have direct geometric meaning as morphisms of schemes.  Thus the papers~\cite{MR0232821,MR1615928} do not reconstruct a scheme \emph{functorially} from its category of modules.  For comparison, the papers~\cite{Lur09,Brandenburg2011,MR3097055,MR3144607,MR4009673,MR3361309} do provide functorial reconstruction of various algebrogeometric objects by remembering their module categories' symmetric monoidal structures and demanding that functors be symmetric monoidal.
  
By a similar token, an associative algebra $A$ is not determined by the equivalence type of the category $\cat{Mod}_{A}$, which encodes only the class of $A$ in the Morita bicategory $\cat{Mor}$.  But $A$ can be recovered if $\cat{Mod}_{A}$ is equipped with a \define{pointing} --- a distinguished object --- in $\cat{Mod}_{A}$, namely the rank-one free module $A_{A}$.  I therefore propose:

\begin{definition}\label{defn.noncomstack}
  A \define{noncommutative 1-affine stack $X$ over $\bK$} is a $\bK$-linear cocomplete category $\cat{Qcoh}({X})$ equipped with a distinguished object $\unit_X = \cO_{X} \in \cat{Qcoh}(X)$.  
  
  Let $\cat{Cocomp}_{\bK}$ denote the bicategory whose objects are $\bK$-linear cocomplete categories, whose 1-morphisms are cocontinuous $\bK$-linear functors, and whose 2-morphisms are natural transformations.  It is symmetric monoidal for a version of Deligne's tensor product $\boxtimes$; the monoidal unit is $\cat{Vect}$ \cite[Section 6.5]{MR2177301}.  Noncommutative 1-affine stacks are the objects of a bicategory which we will call $\cat{Alg}_{0}^{\mathrm{lax}}(\cat{Cocomp}_{\bK})$.  It has:
  \begin{description}
    \item[Objects] An object of  $\cat{Alg}_{0}^{\mathrm{lax}}(\cat{Cocomp}_{\bK})$ is a noncommutative stack, i.e.\ a pair $(\cC,\unit_\cC\in \cC)$ where $\cC \in \cat{Cocomp}_{\bK}$ is a cocomplete $\bK$-linear category and $\unit_\cC$ is a pointing thereof.
    \item[1-morphisms] A 1-morphism $(\cA,\unit_\cA) \to (\cB,\unit_\cB)$ is a pair $(F,\unit_F)$ where $F : \cA \to \cB$ is a $\bK$-linear cocontinuous functor and $1_F : \unit_\cB \to F(\unit_\cA)$ is a homomorphism in $\cB$.
    \item[2-morphisms] A 2-morphism $(F,1_F) \to (G,1_G)$ is a natural transformation $\eta : F \to G$ such that $\eta_{\unit_\cA} \circ 1_F = 1_G : \unit_\cB \to G(\unit_\cA)$.
  \end{description}
  The bicategory $\cat{Alg}_{0}^{\mathrm{lax}}(\cat{Cocomp}_{\bK})$ is symmetric monoidal for $\boxtimes$.
\end{definition}

\begin{remark}\label{set theoretic difficulties}
Recall that a category is \define{cocomplete} if it is closed under colimits, and a functor is \define{cocontinuous} if it preserves colimits.  For set theoretic reasons it is often preferable to work just with locally presentable categories rather than all cocomplete categories.  Actually, \cite[Section 6.5]{MR2177301} works with small categories closed under some small set of colimit shapes, and so to extract the tensor product $\boxtimes$ on $\cat{Cocomp}_\bK$ requires the type of judicious Grothendieck-universe jumping standard in category theory.
\end{remark}

\begin{remark}\label{lax pointed map}
  The 1-morphisms $(F,1_F) : (\cA,\unit_\cA) \to (\cB,\unit_\cB)$ in $\cat{Alg}_{0}^{\mathrm{lax}}(\cat{Cocomp}_{\bK})$ are \define{lax homomorphisms of pointed categories}, a.k.a.\ \define{lax pointed functors}, as opposed to the \define{strong} pointed functors of \cref{remark.strongpointedfunctor}. 
   The use of ``lax'' here is consistent with the general notion of ``lax homomorphism'' in~\cite{JFS}.
\end{remark}

\begin{example}\label{eg.EW}
  The Eilenberg--Watts theorem~\cite{MR0125148,MR0118757} asserts that the functor $\cat{Mor} \to \cat{Cocomp}_{\bK}$ sending an algebra $A$ to the category $\cat{Mod}_{A}$ of right $A$-modules and a bimodule $_{A}M_{B}$ to the functor $(-) \otimes_{A} M : \cat{Mod}_{A} \to \cat{Mod}_{B}$ is a fully faithful inclusion of bicategories in the sense that it induces an equivalence of categories
  $$ \hom_{\cat{Cocomp}_{\bK}}(\cat{Mod}_{A},\cat{Mod}_{B}) \cong \hom_{\cat{Mor}}(A,B).$$
  
  Given an algebra $A$, consider the pointed category $(\cat{Mod}_{A},A_{A}) \in \cat{Alg}_{0}^{\mathrm{lax}}(\cat{Cocomp}_{\bK})$.  What is the category of homomorphisms $(F,f) : (\cat{Mod}_{A},A_{A}) \to (\cat{Mod}_{B},B_{B})$?  By the Eilenberg--Watts theorem, the data of a cocontinuous linear functor $F : \cat{Mod}_{A} \to \cat{Mod}_{B}$ consists (up to canonical isomorphism) of an $A$-$B$-bimodule $_{A}F_{B}$.  What about the homomorphism $f : B_{B} \to F(A) \cong {A_{A}} \otimes_{A} {_{A}F_{B}}$?  It is nothing but an element of the underlying vector space of $F$.
  
  Thus the Eilenberg--Watts inclusion $\cat{Mor} \mono \cat{Cocomp}_{\bK}$ lifts to an inclusion
  $$ \EW: \cat{Alg}_{1}(\cat{Vect}_{\bK}) \mono \cat{Alg}_{0}^{\mathrm{lax}}(\cat{Cocomp}_{\bK})$$ sending $ A \mapsto (\cat{Mod}_{A},A_{A}). $   It is reasonable to declare therefore that a non-commutative stack is \define{(0-)affine} if it is in the essential image of $\EW$.  
  
  It is worth noting that $\EW$ is symmetric monoidal.
\end{example}

\begin{remark}
  The composition $\EW \circ \cM : \cat{Alg}^\mathrm{homo} \to \cat{Alg}_{0}^{\mathrm{lax}}(\cat{Cocomp}_{\bK})$ of the Eilenberg--Watts and modulation functors picks out precisely those 1-morphisms $(F,1_F) : (\cA,\unit_\cA) \to (\cB,\unit_\cB)$ which are \define{strong} pointed functors in the sense that $1_F : \unit_\cB \to F(\unit_\cA)$ is an isomorphism.
\end{remark}

With \cref{defn.noncomstack} in place, we may study quantum field theories valued in ``noncommutative stacks'':

\begin{definition}\label{defn.nonaffineqft}
  A \define{(non-extended) 1-affine Heisenberg-picture quantum field theory} over~$\bK$ is a contravariant symmetric monoidal functor $\cat{Spacetimes} \to \cat{Alg}_{0}^{\mathrm{lax}}(\cat{Cocomp}_{\bK})$.
\end{definition}

Any 0-affine quantum field theory $\cat{Spacetimes} \to \cat{Alg}_{1}(\cat{Vect}_{\bK})$ provides an example of a 1-affine quantum field theory,
 by composing with the Eilenberg--Watts inclusion from \cref{eg.EW}.

\begin{remark}
  One-affine Heisenberg-picture quantum field theory is one possible formalization of \define{twisted}~\cite{MR2742432} or \define{relative}~\cite{MR3165462} quantum field theory: Schr\"odinger-picture quantum field theory valued in a categorified Schr\"odinger-picture quantum field theory.  Indeed, given a Heisenberg-picture quantum field theory $Z : \cat{Spacetimes} \to \cat{Alg}_{0}^{\mathrm{lax}}(\cat{Cocomp}_{\bK})$, one can forget  to a ``categorified'' quantum field theory $Z_{0} : \cat{Spacetimes} \to \cat{Cocomp}_{\bK}$.  The extra data of the field theory $Z$ is a \define{symmetric monoidal oplax natural transformation} from the trivial field theory to $Z_{0}$ in the sense of~\cite{JFS}.
\end{remark}

\begin{example}
  Given a vector space $V \in \cat{Vect}_{\bK}$, there are two natural ways to produce an object of $\cat{Alg}_{0}^{\mathrm{lax}}(\cat{Cocomp}_{\bK})$:
  \begin{enumerate}
  \item Following \cref{prop.end} and \cref{eg.EW}, apply $\End : \cat{Vect} \to \cat{Alg}_{1}(\cat{Vect}_{\bK})$ to produce the algebra $A = \End(V)$ and then apply $\EW : \cat{Alg}_{1}(\cat{Vect}_{\bK}) \mono \cat{Alg}_{0}^{\mathrm{lax}}(\cat{Cocomp}_{\bK})$ to produce the pointed category $\bigl(\cat{Mod}_{\End(V)},\End(V)_{\End(V)}\bigr)$.
  \item  Recognize that $(\cat{Vect}_{\bK},V)$ is already a pointed category, pointed by $V$ rather than $\bK$.
  \end{enumerate}
  Both constructions are contravariantly functorial.  Functoriality of the first we have already addressed.  For the second, given a linear map $f: V \to W$, the pair $(\id_{\cat{Vect}},f)$ is a pointed functor $(\cat{Vect}_{\bK},W) \to (\cat{Vect}_{\bK},V)$.  (The Eilenberg--Watts theorem identifies cocontinuous functors $\cat{Vect} \to \cat{Vect}$ with vector spaces.  So a general pointed functor $(\cat{Vect}_{\bK},V) \to (\cat{Vect}_{\bK},W)$ consists of a pair $(X,f)$ where $X \in \cat{Vect}$ is a vector space and $f \in \hom(W,V\otimes X)$ is a linear map.)
  
  But these constructions are not honestly different for the most important $V$.  Suppose that $V \neq 0$ is finite-dimensional.  Then $\End(V)$ is Morita-equivalent to $\bK$: there is an equivalence of categories $\cat{Mod}_{\End(V)}\simeq \cat{Vect}_{\bK}$.  The choice of $V$ provides a canonical such equivalence $\otimes_{\End(V)}V$.  Under this equivalence, the object $V^{*}_{\End(V)} \in \cat{Mod}_{\End(V)}$ is identified with $\bK \in \cat{Vect}_{\bK}$ and the rank-one free module $\End(V)_{\End(V)} \in \cat{Mod}_{\End(V)}$ is identified with $V \in \cat{Vect}_{\bK}$.  So, at least for non-zero  finite-dimensional $V$, the pointed categories $\bigl(\cat{Mod}_{\End(V)},\End(V)_{\End(V)}\bigr)$ and $(\cat{Vect}_{\bK},V)$ are equivalent in $\cat{Alg}_{0}^{\mathrm{lax}}(\cat{Cocomp}_{\bK})$.  Similar remarks apply when $V$ is a separable Hilbert space and one works in an analytic context in which the algebra of bounded operators on $V$ is Morita-equivalent to $\bK$.
  
  The construction $V \mapsto (\cat{Vect}_{\bK},V)$ is fully symmetric monoidal whereas, as mentioned in \cref{prop.end}, $V \mapsto \End(V)$ is only symmetric monoidal when $\dim V < \infty$.  So in some sense construction~(2) above is the more natural one: it turns \emph{any} Schr\"odinger-picture quantum field theory, independent of dimension, into a Heisenberg-picture one.  But it also explains construction~(1): most Schr\"odinger-picture quantum field theories (including all topological ones) are valued in vector spaces~$V$ for which $(\cat{Vect}_{\bK},V)$ is {(0-)affine} in the sense of \cref{eg.EW}.
\end{example}

\section{Extended affine field theory}

Atiyah~\cite{MR1001453} and Segal~\cite{MR2079383} introduced their functorial axioms for quantum field theory in an attempt to capture the \emph{locality} of physics.  In the decades since, it has become clear that locality is stronger than a functor that takes values on $(d-1)$-dimensional ``spaces'' and $d$-dimensional ``spacetimes'': a quantum field theory should also assign algebraic data to $k$-dimensional manifolds for lower $k$, and to spacetimes with such corners, since manifolds can be cut and glued along such manifolds \cite{MR1273575,MR1256993,BaeDol95}.  

The modern consensus (see e.g.~\cite{Lur09,ClaudiaThesis}) is that to fully capture locality, the source $\cat{Spacetimes}_d$ of a $d$-dimensional quantum field theory should not be just a (symmetric monoidal) category, but a (symmetric monoidal) $(\infty,d)$-category: $k$-dimensional morphisms for $k\leq d$ should be $k$-dimensional manifolds with corners, equipped with the appropriate geometric structure; ``higher'' morphisms for $k > d$ should be isomorphisms of spacetimes and isotopies (of isotopies of \dots)\ thereof.
The target of a quantum field theory must also be an $(\infty,d)$-category.  

\begin{definition}\label{defn.fullyextendedSch}
  A \define{delooping} of a symmetric monoidal $(\infty,n)$-category $\cC$ is a choice of $(\infty,n+1)$-category $\cD$ along with an equivalence $\cC \cong \hom(\unit_{\cD},\unit_{\cD})$, where $\unit_{\cD}$ is the unit object in $\cD$.
  
  A \define{$d$-dimensional fully-extended Schr\"odinger-picture quantum field theory over $\bK$} is a symmetric monoidal functor $\cat{Spacetimes}_d^\cG \to d\cat{Vect}$, where:
  \begin{itemize}
  \item $\cat{Spacetimes}_d^\cG$ is some symmetric monoidal $(\infty,d)$-category whose $d$-morphisms are $d$-dimensional cobordisms with corners equipped with geometric structure of type $\cG$;
  \item $d\cat{Vect}_{\bK}$ is some $(d-1)$-fold delooping of $\cat{Vect}_{\bK}$.
  \end{itemize}
  For example, the bicategories $\cat{Mor}_{\bK}$ and $\cat{Cocomp}_{\bK}$ are reasonable choices for $\cat{2Vect}_{\bK}$.
\end{definition}

\begin{remark} \label{cutting}
\Cref{defn.fullyextendedSch} is incomplete because, depending on the geometry $\cG$, constructing an $(\infty,d)$-category deserving the name $\cat{Spacetimes}_d^\cG$ may be quite difficult. When $\cG$ is sufficiently topological, $\cat{Spacetimes}_d^\cG$ can be built from the ``topological'' bordism category $\cat{Bord}^{\mathrm{smooth}}_d$ following \cite[Section 3.2]{Lur09} or \cite[Section 3.3]{Schommer-Pries:thesis} (elaborated upon in \cite{MR3623677}).  A detailed outline of the construction of $\cat{Bord}^{\mathrm{smooth}}_d$ itself is given in \cite{Lur09} and clarified in \cite{MR3924174}.

  Although the modern consensus is that $\cat{Spacetimes}_d^\cG$ should be a symmetric monoidal higher category, there is some evidence that the usual notions of ``higher category'' (e.g.\ those in \cite{BSP2011}) may not provide the correct framework in which to organize cobordisms, and that other related versions are needed~\cite{MR2806651}.  
\end{remark}

For a similar generalization of Heisenberg-picture quantum field theories, we should look for interesting deloopings of $\cat{Alg}_{1}(\cat{Vect}_{\bK})$ or $\cat{Alg}_{0}^{\mathrm{lax}}(\cat{Cocomp}_{\bK})$.  Various options are available, but we will use one which has a natural interpretation in terms of the pictures drawn in quantum field theory of insertion of local observables.  To wit, consider choosing a vector space~$V$ and placing on~$\bR$ (considered just as an oriented manifold) some ``beads'' labeled by vectors in~$V$.  These ``beads'' can slide back and forth but cannot pass through each other.  Further assume that the beads can ``fuse'' in a nuclear reaction.  Imposing linearity in the labels, this ``fusion'' is an operation $V \otimes V \to V$.  At microscopic scales, fusion isn't really a discrete process: instead, when beads become very close to each other, they can become bonded and behave like a single particle.  But suppose  that all physics of beads and fusion is ``topological'' in the sense that it is independent of distances on~$\bR$.  Then the fusion operation $V \otimes V \to V$ is associative.  The ``invisible bead'' ---  no bead at all --- is a unit for this fusion.  So such a system is the same as an associative algebra structure on~$V$.  The beads are nothing but ``local observables'' drawn from the ``algebra of observables''~$V$.  The theory of local observables for a general quantum field theory has been formulated in terms of \define{factorization algebras} in~\cite{CG}, and restricts to ``beads on $\bR$'' in the case of one-dimensional topological theories.

More generally, consider dividing~$\bR$ into intervals separated by ``point defects.''  One can imagine a situation in which the different regions follow different physical laws: one might be filled with water, for example, and another air.  Each defect might also have its own physics: perhaps there are waves that live only where water and air meet.  Each region and each defect has its own vector space of observables.  As observables move around, they can fuse, and we will require that the universe be topological.  An observable in a region can move onto a defect, but not otherwise.  There should be ``invisible observables'' that can be inserted at any point.  These rules together comprise (1) an associative algebra assigned to each region, and (2) a pointed bimodule assigned to each defect.  These are nothing but the objects and morphisms of $\cat{Alg}_{1}(\cat{Vect}_{\bK})$.

The natural generalization is to consider systems of regions and defects on higher-dimensional spaces --- vector spaces of observables assigned to each stratum, with operations describing the ways points can collide.  Let's impose a topological condition.  The Eckmann--Hilton argument~\cite{MR0136642} then says that the algebra assigned to any two-dimensional region is commutative:
\begin{multline*}
\begin{tikzpicture}[baseline=(center)]
 \draw[fill=black!15] (0,-3pt) coordinate (center) (-1,-1) rectangle (1,1);
 \path (0,-3pt) node[draw,circle,inner sep=1pt,fill=white,anchor=base] {$vw$};
\end{tikzpicture}
\quad \overset{\text{fusion}}= \quad
\begin{tikzpicture}[baseline=(center)]
 \draw[fill=black!15] (0,-3pt) coordinate (center) (-1,-1) rectangle (1,1);
 \path (0:-.5) ++(0,-3pt) node[draw,circle,inner sep=1pt,fill=white,anchor=base] {$v$};
 \path (0:.5) ++(0,-3pt) node[draw,circle,inner sep=1pt,fill=white,anchor=base] {$w$};
\end{tikzpicture}
\quad = \quad
\begin{tikzpicture}[baseline=(center)]
 \draw[fill=black!15] (0,-3pt) coordinate (center) (-1,-1) rectangle (1,1);
 \path (60:-.5) ++(0,-3pt) node[draw,circle,inner sep=1pt,fill=white,anchor=base] {$v$};
 \path (60:.5) ++(0,-3pt) node[draw,circle,inner sep=1pt,fill=white,anchor=base] {$w$};
\end{tikzpicture}
\\ = \quad
\begin{tikzpicture}[baseline=(center)]
 \draw[fill=black!15] (0,-3pt) coordinate (center) (-1,-1) rectangle (1,1);
 \path (120:-.5) ++(0,-3pt) node[draw,circle,inner sep=1pt,fill=white,anchor=base] {$v$};
 \path (120:.5) ++(0,-3pt) node[draw,circle,inner sep=1pt,fill=white,anchor=base] {$w$};
\end{tikzpicture}
\quad = \quad
\begin{tikzpicture}[baseline=(center)]
 \draw[fill=black!15] (0,-3pt) coordinate (center) (-1,-1) rectangle (1,1);
 \path (180:-.5) ++(0,-3pt) node[draw,circle,inner sep=1pt,fill=white,anchor=base] {$v$};
 \path (180:.5) ++(0,-3pt) node[draw,circle,inner sep=1pt,fill=white,anchor=base] {$w$};
\end{tikzpicture}
\quad \overset{\text{fusion}}= \quad
\begin{tikzpicture}[baseline=(center)]
 \draw[fill=black!15] (0,-3pt) coordinate (center) (-1,-1) rectangle (1,1);
 \path (0,-3pt) node[draw,circle,inner sep=1pt,fill=white,anchor=base] {$wv$};
\end{tikzpicture}
\end{multline*}
What about a one-dimensional defect separating regions whose commutative algebras of functions are  $A$ and $B$?  Its local observables form an associative $(A \otimes B)$-algebra.

Commutative algebras are mildly disappointing because they are not particularly ``quantum.''  Fortunately, there is more to the Eckmann--Hilton argument than the answer ``the algebra is commutative.''  Indeed, the Eckmann--Hilton argument actually gives \emph{two} proofs of the commutativity of multiplication, the above one and:
\begin{multline*}
\begin{tikzpicture}[baseline=(center)]
 \draw[fill=black!15] (0,-3pt) coordinate (center) (-1,-1) rectangle (1,1);
 \path (0,-3pt) node[draw,circle,inner sep=1pt,fill=white,anchor=base] {$vw$};
\end{tikzpicture}
\quad \overset{\text{fusion}}= \quad
\begin{tikzpicture}[baseline=(center)]
 \draw[fill=black!15] (0,-3pt) coordinate (center) (-1,-1) rectangle (1,1);
 \path (0:-.5) ++(0,-3pt) node[draw,circle,inner sep=1pt,fill=white,anchor=base] {$v$};
 \path (0:.5) ++(0,-3pt) node[draw,circle,inner sep=1pt,fill=white,anchor=base] {$w$};
\end{tikzpicture}
\quad = \quad
\begin{tikzpicture}[baseline=(center)]
 \draw[fill=black!15] (0,-3pt) coordinate (center) (-1,-1) rectangle (1,1);
 \path (-60:-.5) ++(0,-3pt) node[draw,circle,inner sep=1pt,fill=white,anchor=base] {$v$};
 \path (-60:.5) ++(0,-3pt) node[draw,circle,inner sep=1pt,fill=white,anchor=base] {$w$};
\end{tikzpicture}
\\ = \quad
\begin{tikzpicture}[baseline=(center)]
 \draw[fill=black!15] (0,-3pt) coordinate (center) (-1,-1) rectangle (1,1);
 \path (-120:-.5) ++(0,-3pt) node[draw,circle,inner sep=1pt,fill=white,anchor=base] {$v$};
 \path (-120:.5) ++(0,-3pt) node[draw,circle,inner sep=1pt,fill=white,anchor=base] {$w$};
\end{tikzpicture}
\quad = \quad
\begin{tikzpicture}[baseline=(center)]
 \draw[fill=black!15] (0,-3pt) coordinate (center) (-1,-1) rectangle (1,1);
 \path (180:-.5) ++(0,-3pt) node[draw,circle,inner sep=1pt,fill=white,anchor=base] {$v$};
 \path (180:.5) ++(0,-3pt) node[draw,circle,inner sep=1pt,fill=white,anchor=base] {$w$};
\end{tikzpicture}
\quad \overset{\text{fusion}}= \quad
\begin{tikzpicture}[baseline=(center)]
 \draw[fill=black!15] (0,-3pt) coordinate (center) (-1,-1) rectangle (1,1);
 \path (0,-3pt) node[draw,circle,inner sep=1pt,fill=white,anchor=base] {$wv$};
\end{tikzpicture}
\end{multline*}
Linear maps are either equal or unequal, but in a homotopical or higher categorical world, \emph{how} two operations are ``the same'' is a type of data.  One can therefore define a nontrivial notion of \define{$n$-algebra} in a symmetric monoidal category $\cS$ to be an object of $\cS$ with operations parameterized by labeled configurations of points in $\bR^{n}$ such that homotopies between configurations correspond to homotopies between operations.  (A precise definition, along with much discussion, is in \cite[Chapter 5]{HA}, where $n$-algebras are called \define{$\bE_{n}$-algebras}.)  For example, a $1$-algebra is a homotopy-coherent version of an associative algebra.  A $2$-algebra in the $(\infty,1)$-category of (usual) categories is a braided monoidal category.  A $0$-algebra is a \define{pointed object}, i.e.\ an object $X\in \cS$ along with a map $\unit \to X$, where $\unit\in \cS$ is the monoidal unit.

The bicategory $\cat{Alg}_{1}(\cat{Vect}_{\bK})$ of associative algebras and pointed bimodules can then be generalized to higher algebra by working, as discussed above, with systems of algebras of observables that are topological but vary at prescribed lower-dimensional strata in $\bR^{n}$.  The following summarizes the main results of~\cite{ClaudiaThesis}:

\begin{theorem}[\cite{ClaudiaThesis}] \label{thm.Scheimbauer}
  Let $\cS$ be a symmetric monoidal $(\infty,1)$-category admitting filtered colimits and such that the symmetric monoidal structure distributes over filtered colimits.  
  \begin{enumerate}
  \item For each $n \in \bN$, there is a symmetric monoidal $(\infty,n)$-category $\cat{Alg}_{n}(\cS)$ with:
  \begin{description}
    \item[Objects] $n$-algebras in $\cS$.
    \item[1-morphisms] $(n-1)$-algebra bimodules between $n$-algebras.
    \item[2-morphisms] $(n-2)$-algebra bimodules between $(n-1)$-algebras, for which the various actions of the ambient $n$-algebras are compatible.
    \item[\dots] \dots
    \item[$n$-morphisms] pointed bimodules between $1$-algebras, for which the various actions of lower-dimensional morphisms are compatible.
    \item[$(n+1)$-morphisms] equivalences of $n$-morphisms.
    \item[$(n+2)$-morphisms] equivalences of equivalences.
    \item[\dots] \dots
  \end{description}
  \item Two $n$-algebras in $\cS$ are equivalent as objects in $\cat{Alg}_{n}(\cS)$ if and only if they are equivalent as $n$-algebras (i.e.\ via homomorphisms, rather than via bimodules).
  \item Let $\cat{Bord}^{\mathrm{fr}}_{d}$ denote the ``spacetime'' $(\infty,d)$-category of framed smooth cobordisms.  For each $n$-algebra $X$ in $\cS$, there is a unique (up to contractible choices) symmetric monoidal functor $\cat{Bord}^{\mathrm{fr}}_{d} \to \cat{Alg}_{n}(\cS)$ assigning $X$ to the standard-framed point $\pt \in \cat{Bord}^{\mathrm{fr}}_{d}$.  This functor is called \define{factorization homology} or \define{topological chiral homology} with coefficients in $X$, written $M \mapsto \int_{M}X$. \qedhere
  \end{enumerate}
\end{theorem}

\begin{remark}
 Part (3) of \cref{thm.Scheimbauer} is proved by explicitly constructing the functor $\int_{\Box} X$, rather than by appealing to Lurie's celebrated classification theorem from~\cite{Lur09}.  Lurie called~\cite{Lur09} an ``outline,'' and while it seems the consensus among experts is that it is in every important way correct and complete, it also seems best to prove results without appealing to an ``outline.''
\end{remark}

As I said earlier, it is mildly disappointing that $2$-algebras in $\cat{Vect}$ are automatically commutative.  But there is a close cousin to $\cat{Vect}$ that admits honestly non-commutative $n$-algebras for all $n$.  The category $\cat{DGVect}$ of chain complexes of vector spaces (``derived'' vector spaces) has a model category structure making it into an $(\infty,1)$-category in an interesting way: objects are chain complexes and 1-morphisms are chain maps as in the usual category, but 2-morphisms are chain homotopies between 1-morphisms, 3-morphisms are homotopies between homotopies, and so on.  It satisfies the conditions of \cref{thm.Scheimbauer}.  We can therefore define:

\begin{definition}
  A \define{$d$-dimensional fully-extended derived affine Heisenberg-picture quantum field theory} based on the spacetime category $\cat{Spacetimes}_d^\cG$ is a contravariant symmetric monoidal functor of $(\infty,d)$-categories $\cat{Spacetimes}_d^\cG \to \cat{Alg}_{d}(\cat{DGVect})$.
\end{definition}

Then part (3) of \cref{thm.Scheimbauer} asserts that each $d$-algebra (among chain complexes) defines a $d$-dimensional fully-extended framed topological derived affine Heisenberg-picture quantum field theory.

\section{Extended non-affine field theory}

Of course,  category theory  accommodates affine \emph{non}-derived quantum field theories --- functors to $\cat{Alg}_{d}(\cat{Vect})$ --- but these will not exhibit truly ``quantum'' behavior except in codimensions~$0$ and~$1$.  If the derived world is to be avoided, another option for capturing honestly ``quantum'' examples is to give up on affineness.  (Indeed, as we will see in \cref{section.skein theory}, important examples are not affine.)  Let $\cS$ be a symmetric monoidal $(\infty,k)$-category, say the bicategory $\cat{Cocomp}_{\bK}$.  The notion of ``$d$-algebra in $\cS$'' never uses non-invertible $2$- or higher morphisms, and so depends only on the maximal $(\infty,1)$-category inside of $\cS$.  Thus one can throw away that data and define $\cat{Alg}_{d}(\cS)$ as in \cref{thm.Scheimbauer}.  But the higher morphisms in $\cS$ can be used to enriched $\cat{Alg}_{d}(\cS)$:

\begin{theorem}[\cite{JFS}] \label{higher alg n}
  Let $\cS$ be a symmetric monoidal $(\infty,k+1)$-category satisfying conditions analogous to those of \cref{thm.Scheimbauer} (for the details, see \cite{JFS}).  Then the $(\infty,n)$-category $\cat{Alg}_{n}(\cS)$ from \cref{thm.Scheimbauer} is the $n$-dimensional truncation of an $(\infty,n+k+1)$-category $\cat{Alg}_{n}^{\mathrm{lax}}(\cS)$.  This extended version has the same $0$- through $n$-morphisms as its non-extended cousin.  
  In particular, the $n$-morphisms are pointed objects of $\cS$ which are acted upon by their sources and targets.  The $(n+1)$-morphisms are lax homomorphisms of pointed bimodules; the $(n+2)$-morphisms are lax homomorphisms thereof; etc.
   \qedhere
\end{theorem}

\begin{examplenodiamond}\label{alg0S}
  Let $\cS$ be a symmetric monoidal $(\infty,k+1)$-category with unit object $\unit \in \cS$.  Generalizing \cref{defn.noncomstack}, the $(\infty,k+1)$-category $\cat{Alg}_0^{\mathrm{lax}}(\cS)$ has:
  \begin{description}
  \item[Objects] An object of $\cat{Alg}_0^{\mathrm{lax}}(\cS)$ consists of a pair $(X,1_X)$ where object $X \in \cS$ and $1_X : \unit \to X$ is a \define{pointing} of $X$.
  $$ \begin{tikzpicture}[anchor=base,auto,baseline=(K.base)]
  \path (0,0) node (K) {$\unit$} (3,0) node (X) {$X$} ;
  \draw[->] (K) -- node (x) {$1_X$} (X);
\end{tikzpicture} $$
  \item[1-morphism] A 1-morphism $(X,1_X) \to (Y,1_Y)$ is a \define{lax homomorphism} of pointed objects.  It consists of a pair $(F,1_F)$ where $F : X \to Y$ is a 1-morphism in $\cS$ and $1_F : 1_Y \to F\circ 1_X$ is a 2-morphism.
  $$ \begin{tikzpicture}[anchor=base,auto,baseline=(K.base)]
  \path node (K) {$\unit$} +(3,1) node (X) {$X$} +(3,-1) node (Y) {$Y$} ;
  \draw[->] (K) -- node (x) {$1_X$} (X);
  \draw[->] (K) -- node[swap] (y) {$1_Y$} (Y);
  \draw[->] (X) -- node (F) {$F$} (Y);
  \draw[twoarrowshorter] (K) -- node {$\scriptstyle 1_F$} (F);
\end{tikzpicture} $$
  \item[2-morphism] A 2-morphism $(F,1_F) \to (G,1_G)$ is pair $(T,1_T)$ where $T : F \to G$ is a 2-morphism in $\cS$ and $1_T : 1_G \to T\circ 1_F$ is a 3-morphism.
  $$ \begin{tikzpicture}[anchor=base,auto,baseline=(K.base)]
  \path node (K) {$\unit$} +(4,2) node (X) {$X$} +(4,-2) node (Y) {$Y$} +(2.75,0) node {$\overset {1_T} \Rrightarrow$};
  \draw[->,thick] (K) -- node (x) {$1_X$} (X);
  \draw[->,thick] (K) -- node[swap] (y) {$1_Y$} (Y);
  \draw[->,thick] (X) .. controls +(1,-1) and +(1,1).. node[pos=.35] (G) {$G$} (Y);
  \draw[twoarrow,thin] (K) -- node {$\scriptstyle 1_G$} (G);
  \draw[->,thick] (X) .. controls +(-1,-1) and +(-1,1).. node[pos=.55,swap] (F) {$F$} (Y);
  \draw[twoarrowlonger,thick] (K) -- node[pos=.75,swap] {$\scriptstyle 1_F$} (F);
  \draw[twoarrow,thick] (F) -- node[swap] {$ T$} (G);
\end{tikzpicture} $$
  \item[\dots] \dots \diamondhere
  \end{description}
\end{examplenodiamond}

\begin{example}\label{alg1cocomp}
Up to a contractible space of choices, one can identify objects of $\cat{Alg}_1^{\mathrm{lax}}(\cat{Cocomp}_\bK)$ with \define{monoidal} $\bK$-linear cocomplete  categories $\cC$, by which I mean that the monoidal functor is $\bK$-linear and cocontinuous in each variable, so that it extends to a functor $\otimes : \cC \boxtimes_\bK \cC \to \cC$.  Recall that to be monoidal, in addition to the monoidal functor, we should have distinguished ``associator'' and ``unitor'' natural isomorphisms satisfying standard ``pentagon'' and ``triangle'' axioms.  I will generally suppress these auxiliary data.

Let $\cC$ be a {monoidal} $\bK$-linear cocomplete  category and $\cX$ a $\bK$-linear cocomplete  category.  A \define{left action} of $\cC$ on $\cX$ consists of a 1-morphism (i.e.\ $\bK$-linear cocontinuous functor) $\triangleright : \cC \boxtimes_\bK \cX \to \cX$ and an associator $(C_1\otimes C_2)\triangleright X \isom C_1 \triangleright (C_2 \triangleright X)$ and a unitor $\unit_\cC \triangleright X \isom X$ (natural in $X\in \cX$ and $C_1,C_2 \in \cC$, of course) satisfying the appropriate pentagon and triangle equations.  A \define{right action}  is similar: there should be a functor $\triangleleft : \cX \boxtimes_\bK \cC \to \cX$ and an associator and a unitor.  Suppose that $(\cB,\otimes_\cB)$ and $(\cC,\otimes_\cC)$ are monoidal $\bK$-linear cocomplete category and $\cX$ is a $\bK$-linear cocomplete category equipped with a left action $\triangleright : \cB \boxtimes_\bK \cX \to \cX$ and a right action $\cX \boxtimes_\bK \cC \to \cX$.  A \define{compatibility} between these actions consists of an associator $\alpha_{B,X,C} : (B \triangleright X) \triangleleft C \isom B\triangleright (X \triangleleft C)$ that satisfies the appropriate pentagon and triangle equations.  Such an $\cX$ is called a \define{$\cB$-$\cC$-bimodule}.

A {1-morphism} in $\cat{Alg}_1^{\mathrm{lax}}(\cat{Cocomp}_\bK)$ between {monoidal} $\bK$-linear cocomplete  categories $\cB$ and $\cC$ is a \define{pointed} $\cB$-$\cC$-bimodule, i.e.\ a $\cB$-$\cC$-bimodule $\cX$ with a distinguished object $\unit_\cX \in \cX$.  Let $(\cX,\unit_\cX)$ be a pointed $\cB$-$\cC$-bimodule and $(\cY,\unit_\cY)$ a pointed $\cC$-$\cD$-bimodule.  The \define{balanced tensor product} $\cX \boxtimes_\cC \cY$ is the $\bK$-linear cocomplete category which is universal for the following:
  \begin{itemize}
    \item There is a 1-morphism $P : \cX \boxtimes_\bK \cY \to \cX \boxtimes_\cC \cY$.
    \item There is an ``associator''
    $$ \alpha_{X,C,Y} : P((X\triangleleft C)\boxtimes Y) \isom P(X \boxtimes (C\triangleright Y)) $$
    depending naturally on $X\in \cX$, $C\in \cC$, and $Y \in \cY$.
    \item The associator $\alpha$ satisfies triangle and pentagon equations.  Suppressing the associators and unitors from $\cX,\cC,\cY$, these say that $\alpha_{X,\unit,Y} = \id$ and $\alpha_{X,A\otimes B,Y} = \alpha_{X,A,B\triangleright Y} \circ \alpha_{X\triangleleft A,B,Y}$.
  \end{itemize}
  That the balanced tensor product exists follows from \cite[Theorem 6.23]{MR2177301}, which implies (modulo  \cref{set theoretic difficulties}) that $\cat{Cocomp}_\bK$ contains all colimits.  Indeed, $\cX \boxtimes_\cC \cY$ is the colimit of the following diagram:
    $$
  \begin{tikzpicture}
    \path coordinate (base) +(-.25,0) coordinate (basel) +(.25,0) coordinate (baser);
    \path (base) node (XY) {$\cX \boxtimes \cY$}
      ++(0,2) node (XAY) {$\cX \boxtimes \cC \boxtimes \cY$}
      ++(0,2) node (XAAY) {$\cX \boxtimes \cC \boxtimes \cC \boxtimes \cY$};
    \draw[->,very thick] (XAY) .. controls +(-1,-1) and +(-1,1) .. node[auto,swap] {$\scriptstyle \triangleleft \boxtimes \id$} (XY);
    \draw[->,very thick] (XAY) .. controls +(1,-1) and +(1,1) .. node[auto] {$\scriptstyle \id\boxtimes \triangleright$} (XY);
    \draw[->,very thick] (XAAY) .. controls +(-2,-1) and +(-2,1) .. node[auto,swap] {$\scriptstyle\triangleleft \boxtimes \id \boxtimes \id$} (XAY);
    \draw[->,very thick] (XAAY) .. controls +(2,-1) and +(2,1) .. node[auto] {$\scriptstyle \id\boxtimes\id\boxtimes \triangleright$} (XAY);
    \draw[->,very thick] (XAAY) -- node[fill=white] {$\scriptstyle\id \boxtimes \otimes \boxtimes \id$} (XAY);
    \path[draw,fill=gray,opacity=.25] (basel) .. controls +(-1,1) and +(-1,-1) .. ++(0,2) -- ++(0,2) .. controls +(-2,-1) and +(-2,1) .. ++(-.25,-2) .. controls +(-1,-1) and +(-1,1) .. (basel);
    \path[draw,fill=gray,opacity=.25] (baser) .. controls +(1,1) and +(1,-1) .. ++(0,2) -- ++(0,2) .. controls +(2,-1) and +(2,1) .. ++(.25,-2) .. controls +(1,-1) and +(1,1) .. (baser);
    \path[draw,fill=gray,opacity=.25] (base) .. controls +(-.75,1) and +(-.75,-1) .. ++(0,2) .. controls +(-1.75,1) and +(-1.75,-1) .. ++(0,2) .. controls +(1.75,-1) and +(1.75,1) .. ++(0,-2) .. controls +(.75,-1) and +(.75,1) .. (base);
  \end{tikzpicture}
  $$
  The three 2-cells are:
    \begin{gather*}
  \begin{tikzpicture}[scale=.5,baseline=(XAY)]
    \path coordinate (base) +(-.25,0) coordinate (basel) +(.25,0) coordinate (baser);
    \path (base) node (XY) {}
      ++(0,2) node (XAY) {}
      ++(0,2) node (XAAY) {};
    \draw[->,very thick] (XAY) .. controls +(-1,-1) and +(-1,1) .. node[auto,swap] {} (XY);
    \draw[->,very thick] (XAY) .. controls +(1,-1) and +(1,1) .. node[auto] {} (XY);
    \draw[->,very thick] (XAAY) .. controls +(-2,-1) and +(-2,1) .. node[auto,swap] {} (XAY);
    \draw[->,very thick] (XAAY) .. controls +(2,-1) and +(2,1) .. node[auto] {} (XAY);
    \draw[->,very thick] (XAAY) -- node {} (XAY);
    \path[draw,fill=gray,opacity=.25] (basel) .. controls +(-1,1) and +(-1,-1) .. ++(0,2) -- ++(0,2) .. controls +(-2,-1) and +(-2,1) .. ++(-.25,-2) .. controls +(-1,-1) and +(-1,1) .. (basel);
  \end{tikzpicture}
  \!\!= \text{associator for }\cX,\quad
  \begin{tikzpicture}[scale=.5,baseline=(XAY)]
    \path coordinate (base) +(-.25,0) coordinate (basel) +(.25,0) coordinate (baser);
    \path (base) node (XY) {}
      ++(0,2) node (XAY) {}
      ++(0,2) node (XAAY) {};
    \draw[->,very thick] (XAY) .. controls +(-1,-1) and +(-1,1) .. node[auto,swap] {} (XY);
    \draw[->,very thick] (XAY) .. controls +(1,-1) and +(1,1) .. node[auto] {} (XY);
    \draw[->,very thick] (XAAY) .. controls +(-2,-1) and +(-2,1) .. node[auto,swap] {} (XAY);
    \draw[->,very thick] (XAAY) .. controls +(2,-1) and +(2,1) .. node[auto] {} (XAY);
    \draw[->,very thick] (XAAY) -- node {} (XAY);
    \path[draw,fill=gray,opacity=.25] (baser) .. controls +(1,1) and +(1,-1) .. ++(0,2) -- ++(0,2) .. controls +(2,-1) and +(2,1) .. ++(.25,-2) .. controls +(1,-1) and +(1,1) .. (baser);
  \end{tikzpicture}
  \!\!= \text{associator for }\cY,\\
  \begin{tikzpicture}[scale=.5,baseline=(XAY)]
    \path coordinate (base) +(-.25,0) coordinate (basel) +(.25,0) coordinate (baser);
    \path (base) node (XY) {}
      ++(0,2) node (XAY) {}
      ++(0,2) node (XAAY) {};
    \draw[->,very thick] (XAY) .. controls +(-1,-1) and +(-1,1) .. node[auto,swap] {} (XY);
    \draw[->,very thick] (XAY) .. controls +(1,-1) and +(1,1) .. node[auto] {} (XY);
    \draw[->,very thick] (XAAY) .. controls +(-2,-1) and +(-2,1) .. node[auto,swap] {} (XAY);
    \draw[->,very thick] (XAAY) .. controls +(2,-1) and +(2,1) .. node[auto] {} (XAY);
    \draw[->,very thick] (XAAY) -- node {} (XAY);
    \path[draw,fill=gray,opacity=.25] (base) .. controls +(-.75,1) and +(-.75,-1) .. ++(0,2) .. controls +(-1.75,1) and +(-1.75,-1) .. ++(0,2) .. controls +(1.75,-1) and +(1.75,1) .. ++(0,-2) .. controls +(.75,-1) and +(.75,1) .. (base);
  \end{tikzpicture}
  \!\!=\text{the commutativity }(\id\boxtimes \triangleright) \circ (\triangleleft \boxtimes \id\boxtimes\id) = (\triangleleft\boxtimes\id) \circ (\id \boxtimes\id \boxtimes \triangleright).
  \end{gather*}
  
  The category $\cX \boxtimes_\cC \cY$ is naturally a $\cB$-$\cD$-bimodule, and the balanced tensor product is coherently associative and unital.
  If $\cX$ is pointed by $\unit_\cX$ and $\cY$ by $\unit_\cY$, then $\cX \boxtimes_\cC \cY$ is pointed by $P(\unit_\cX \boxtimes \unit_\cY)$.
This categorifies the 0- and 1-morphisms of the bicategory $\cat{Alg}_1(\cat{Vect}_\bK)$ from \cref{defn.alg}.

Let $(\cX,\unit_\cX)$ and $(\cY,\unit_\cY)$ now be two $\cB$-$\cC$-bimodules.  A {2-morphism} $(\cX,\unit_\cX) \to (\cY,\unit_\cY)$ in $\cat{Alg}_1^{\mathrm{lax}}(\cat{Cocomp}_\bK)$ is a \define{lax pointed bimodule homomorphism}, which is to say a $\bK$-linear cocontinuous functor $F : \cX \to \cY$, natural transformations $B \triangleright F(X) \to F(B \triangleright X)$ and $F(X) \triangleleft C \to F(X \triangleleft C)$ intertwining the various associators and unitors, and a homomorphism $1_F : \unit_\cY \to F(\unit_\cX)$ in $\cY$.

A 3-morphism $(F,1_F) \to (G,1_G)$ is a natural transformation $\eta : F \to G$ intertwining the various natural transformations and homomorphisms.  All together this makes $\cat{Alg}_1^{\mathrm{lax}}(\cat{Cocomp}_\bK)$ into a weak 3-category.
\end{example}

\begin{remark}
I know of no examples of monoidal categories appearing in quantum field theory that are not generated under colimits by their (1-)dualizable objects (see \cref{defn.dualizable}).  Suppose that $\cB$ and $\cC$ are monoidal $\bK$-linear cocomplete categories generated under colimits by their (1-)dualizable objects, and that $\cX$ and $\cY$ are $\cB$-$\cC$-bimodules.  Then in fact all lax bimodule homomorphisms $F : \cX \to \cY$ are \define{strong}, which is when the natural transformations $B \triangleright F(X) \to F(B \triangleright X)$ and $F(X) \triangleleft C \to F(X \triangleleft C)$ are isomorphisms; the proof is the same as that of \cite[Lemma 2.10]{DSPS2}.  The pointing $1_F : \unit_\cY \to F(\unit_\cX)$ is not necessarily an isomorphism. 
\end{remark}

Suppose that $\cS$ is a $k$-fold delooping of $\cat{Vect}$ --- ``the $(\infty,k+1)$-category of cocomplete $\bK$-linear  $(\infty,k)$-categories,'' for example.  Then $\cat{Alg}_n^{\mathrm{lax}}(\cS)$ is an $(n+k)$-fold delooping of $\cat{Alg}_0(\cat{Vect})$; i.e.\ in dimension $(n+k)$, $\cat{Alg}_n^{\mathrm{lax}}(\cS)$ consists of pointed vector spaces.  This is what a fully-extended $(n+k)$-dimensional Heisenberg-picture quantum field theory would assign to top-dimensional spacetimes.  Therefore a \define{$d$-dimensional fully-extended $k$-affine Heisenberg-picture quantum field theory}, for $d\geq k$, should be a symmetric monoidal functor $\cat{Spacetimes}_d \to \cat{Alg}_{d-k}^{\mathrm{lax}}(\cS)$, where $\cS$ is a $k$-fold delooping of $\cat{Vect}$.  For the ``derived'' version, replace $\cat{Vect}$ by $\cat{DGVect}$.  My favorite delooping of $\cat{Vect}_\bK$ is $\cat{Cocomp}_\bK$, and so I generally choose:
\begin{definition}\label{defn.1-affine extended qft}
  A \define{$\cG$-geometric $d$-dimensional fully-extended $1$-affine Heisenberg-picture quantum field theory} is a symmetric monoidal functor $\cat{Spacetimes}_d^\cG \to \cat{Alg}_{d-1}^{\mathrm{lax}}(\cat{Cocomp}_\bK)$, where $\cat{Spacetimes}_d^\cG$ is an $(\infty,d)$-category of $0$- through $d$-dimensional manifolds equipped with  geometric structures of type $\cG$.
\end{definition}

\section{A (non-)dualizability result}

Let us focus on the case of $d$-dimensional fully-extended framed topological quantum field theories, which are based on the spacetime category $\cat{Spacetimes} = \cat{Bord}^{\mathrm{fr}}_{d}$.  The famous \define{cobordism hypothesis}, whose proof is outlined in detail in~\cite{Lur09}, asserts that for any symmetric monoidal $(\infty,N)$-category $\cC$, symmetric monoidal functors $\cat{Bord}^{\mathrm{fr}}_{d} \to \cC$ are the same as $d$-dualizable objects in $\cC$.  (It is not too hard to convince oneself that every such functor picks out a $d$-dualizable object.  The deep part of~\cite{Lur09} is that each $d$-dualizable object determines uniquely such a functor, which is in turn a statement about the topology of manifolds.)

\begin{definition}[\cite{Lur09}]\label{defn.dualizable}
  A 1-morphism $f$ in an $(\infty,N)$-category $\cC$ is \define{1-dualizable} if it has left and right adjoints in the homotopy bicategory $\mathrm h_2\cC$, each of which have left and right adjoints, ad infinitum.  This homotopy bicategory has the same $0$- and $1$-morphisms as has $\cC$, but its $2$-morphisms are the equivalence classes of $2$-morphisms in $\cC$.  (Any data from non-invertible $3$- and higher morphisms in $\cC$ is discarded.)  Thus 1-dualizability asserts the existence of various ``evaluation'' and ``coevaluation'' 2-morphisms in $\cC$, which must satisfy certain relations up to non-unique equivalence.
  
  For $k > 1$, let $f$ be a $k$-morphism in $\cC$ with source and target $(k-1)$-morphisms $x$ and $y$.  Then $x$ and $y$ are \define{parallel}: they have the same source and target $(k-2)$-morphisms.  The morphism $f$ is \define{1-dualizable} if it is $1$-dualizable in the $(\infty,N-k)$-category whose objects are all $(k-1)$-morphisms parallel to $x$ and $y$. Thus $1$-dualizability of a $k$-morphism asserts the existence of various ``evaluation'' and ``coevaluation'' $(k+1)$-morphisms in $\cC$, which must satisfy certain relations up to non-unique equivalence.
  
  For $d > 1$, a $k$-morphism $f$ is \define{$d$-dualizable} if it is $1$-dualizable and the evaluation and coevaluation $(k+1)$-morphisms witnessing such 1-dualizability are themselves $(d-1)$-dualizable.
  
  If $\cC$ is a symmetric monoidal $(\infty,N)$-category, an object $X\in \cC$ is \define{$d$-dualizable} if the functor $\otimes X : \cC \to \cC$ is $d$-dualizable as a $1$-morphism in the $(\infty,N+1)$-category with a single object $\star$ and $\hom(\star,\star) = \cC$ and composition given by $\otimes$.
\end{definition}

In general, given a symmetric monoidal $(\infty,N)$-category $\cC$, it is interesting to ask what are the $d$-dualizable objects for various values of $d$.  First, a trivial observation: 1-dualizable 1-morphisms in an $(\infty,1)$-category are equivalences, and so $(N+1)$-dualizable objects in a symmetric monoidal $(\infty,N)$-category are invertible.  Thus to have non-invertible examples, we should let $d \leq N$.  Well-known examples include:

\begin{example}
  The 1-dualizable objects in $\cat{Vect}$ are the finite-dimensional vector spaces.  The 2-dualizable objects in $\cat{Mor}$ are the finite-dimensional semisimple algebras.
\end{example}

These examples are typical of deloopings of $\cat{Vect}$, in which such ``full'' dualizability is a strong ``finiteness'' condition which nevertheless admits interesting examples.  In terms of field theories, this suggests that $d$-dimensional fully-extended Schr\"odinger field theories are fairly rigid but can be quite nontrivial.  Similar results about $3$-dimensional Schr\"odinger-picture topological field theories are in~\cite{DSPS}.

What about Heisenberg-picture field theories?  \Cref{thm.Scheimbauer} implies that, for any $(\infty,k+1)$-category $\cS$ and any $n$, \emph{every} object of the $(\infty,n+k+1)$-category $\cat{Alg}_n^{\mathrm{lax}}(\cS)$ is $n$-dualizable.  Thus $0$-affine Heisenberg-picture topological field theories (when $n = d$ and $k = 0$) are quite flexible.  But they are not ``fully dualizable'': even when $\cS$ is an $(\infty,1)$-category, $\cat{Alg}_n^{\mathrm{lax}}(\cS)$ is an $(\infty,n+1)$-category.  In fact, asking for any more dualizablity collapses the whole enterprise:

\begin{theorem}[\cite{GwSch18}]\label{theorem about full dualizability}
  Let $\cS$ be a symmetric monoidal $(\infty,k+1)$-category.  Then the groupoid of $(n+k+1)$-dualizable objects in $\cat{Alg}_n^{\mathrm{lax}}(\cS)$ is contractible --- every $(n+k+1)$-dualizable object is canonically equivalent to the unit object $\unit$.
\end{theorem}

\Cref{theorem about full dualizability} arose in conversations joint with Scheimbauer. 
I will give the idea of the proof. Providing all details of the second part of the argument would require working in too much detail for this paper with some particular model of $(\infty,n)$-categories. An alternate complete proof is available in \cite{GwSch18}, which appeared after a preprint of this paper circulated.

\begin{proof}
We begin by proving the claim when $k=0$.
  Let $(X,1_X)$ be a 1-dualizable object in $\cat{Alg}_0^{\mathrm{lax}}(\cS)$ with dual $(X,1_X)^* = (Y,1_Y)$.  Denote the evaluation and coevaluation maps by:
  \begin{gather*}
     (F,1_F) : (\unit,\id_\unit) \to (X,1_X) \otimes (Y,1_Y) \cong (X\otimes Y,1_X\otimes 1_Y), \\
     (G,1_G) : (Y\otimes X,1_Y\otimes 1_X) \cong (Y,1_Y) \otimes (X,1_X) \to (\unit,\id_\unit).
  \end{gather*}
  The compatibility conditions between the evaluation and coevaluation necessary for 1-dualizability assert that two equations, the first of which reads:
  $$ 
  \begin{tikzpicture}[anchor=base,baseline=(K.base)]
  \path 
  node (K) {$\unit$} 
  +(5,3) node (Xtop) {$X$} +(5,-3) node (Xbot) {$X$} +(5,0) node (XYX) {$X\otimes Y\otimes X$} ;
  \draw[->] (K) -- node[auto] (x) {$1_X$} (Xtop);
  \draw[->] (K) -- node[fill=white] (y) {$\scriptstyle 1_X \otimes 1_Y \otimes 1_X$} (XYX);
  \draw[->] (Xtop) -- node[auto] (F) {$F \otimes \id_X   $} (XYX);
  \draw[twoarrowshorter] (K) -- node[fill=white,inner sep=1pt,ellipse,pos=.6] {$ 1_F \otimes \id_{1_X}   $} (F.west);
  \draw[->] (K) -- node[auto,swap] (x) {$1_X$} (Xbot);  
  \draw[->] (XYX) -- node[auto] (F) {$ \id_X \otimes G $} (Xbot);
  \draw[twoarrowshorter] (K) -- node[fill=white,inner sep=1pt,ellipse,pos=.6] {$  \id_{1_X} \otimes 1_G $} (F.west);
\end{tikzpicture}
\quad = \quad
  \begin{tikzpicture}[anchor=base,baseline=(K.base)]
  \path 
  node (K) {$\unit$} 
  +(3,1) node (Xtop) {$X$} +(3,-1) node (Xbot) {$X$} ;
  \draw[->] (K) -- node[auto] (x) {$1_X$} (Xtop);
  \draw[->] (K) -- node[auto,swap] (x) {$1_X$} (Xbot);  
  \draw[->] (Xtop) -- node[auto] (idX) {$\id_X$} (Xbot);
  \draw[twoarrowshorter] (K) -- node[fill=white,inner sep=1pt,pos=.6] {$ \id_{1_X}   $} (idX);
\end{tikzpicture}
 $$
  Considering just $X,Y,F,G$, one finds that $X$ and $Y$ are duals in $\cS$.  Unpacking the extra conditions coming from $1_F$ and $1_G$ gives two commuting triangles, the first of which reads:
  $$
  \begin{tikzpicture}[anchor=base,baseline=(A.base),auto]
    \path node (A) {$1_X$}
       ++(6,0) node (B) {$1_X \circ \bigl(G \circ (1_Y \otimes \id_X)\bigr) \circ 1_X$}
       +(0,1) node (B1) {$1_X \otimes \bigl(G\circ (1_Y \otimes 1_X)\bigr)$}
       ++(6,0) node (C) {$1_X$}
       +(0,1) node (C1) {$(\id_X \otimes G) \circ (F\otimes \id_X) \circ 1_X$};
    \path (B) -- node[rotate=90,anchor=mid] {$\cong$} (B1);
    \path (C) -- node[rotate=90,anchor=mid] {$\cong$} (C1);
    \draw[twoarrowlonger] (A) -- node {$\id_{1_X} \otimes 1_G$} (B);
    \draw[twoarrowlonger] (B) -- node {$1_F \otimes \id_{1_X}$} (C);
    \draw[twoarrowlonger] (A) ..controls +(3,-1.5) and +(-3,-1.5) .. node[swap] {$\id_{1_X}$} (C);
  \end{tikzpicture}
  $$
  These triangles precisely assert that $1_X : \unit \to X$ and $G \circ (1_Y \otimes \id_X) : X \to \unit$ are dual 1-morphisms.  In particular, $1_X$ is 1-dualizable.
  
  An induction implies more generally that a $d$-dualizable object $(X,1_X) \in \cat{Alg}_0^{\mathrm{lax}}(\cS)$ consists of a $d$-dualizable object $X \in \cS$ along with a $d$-dualizable 1-morphism $1_X : \unit \to X$.  But, as remarked above, an $(n+1)$-dualizable 1-morphism in an $(\infty,n+1)$-category is necessarily invertible.  Thus the pointing $1_X$ furnishes a canonical equivalence between $(X,1_X)$ and $(\unit,\id_\unit)$.
  
  Now let $k$ be arbitrary.  I will describe a canonical symmetric monoidal functor $\cat{Alg}_n^{\mathrm{lax}}(\cS) \to \cat{Alg}_0^{\mathrm{lax}}(\cat{Alg}_n^{\mathrm{lax}}(\cS))$ which splits the forgetful map $\cat{Alg}_0^{\mathrm{lax}}(\cat{Alg}_n^{\mathrm{lax}}(\cS)) \to \cat{Alg}_n^{\mathrm{lax}}(\cS)$.  (This splitting is a version of the Eilenberg--Watts functor considered in \cref{eg.EW}.)  With such a functor in hand, any $(n+k+1)$-dualizable object $X \in \cat{Alg}_n^{\mathrm{lax}}(\cS)$ determines an $(n+k+1)$-dualizable object in $\cat{Alg}_0^{\mathrm{lax}}(\cat{Alg}_n^{\mathrm{lax}}(\cS))$, which by above is trivial, and so $X$ is trivial.
  
  First, let $M$ be an $m$-morphism in $\cat{Alg}_n^{\mathrm{lax}}(\cS)$ for $m<n$; i.e.\ an $(n-m)$-algebra in $\cS$, with $n-m\geq 1$, with some compatible actions by some $(n-m+1)$-algebras.  We map it to the pointed $m$-morphism $(M,M_M)$, where $M_M$ is the ``right regular $M$-module,'' which is $M$ treated as an $(n-m-1)$-algebra with the canonical right $M$-action.
In terms of the factorization algebra pictures used in   \cite{ClaudiaThesis} to define $\cat{Alg}_n^{\mathrm{lax}}(\cS)$, the $m$-morphism $M$ is described by a factorization algebra on $\bR^n$ which is locally constant with respect to the stratification given by the subspaces $\{x_1=0\}$, $\{x_1 = x_2 = 0\}$, $\dots$, $\{x_1 = \dots = x_m = 0\}$. 
  The pointed $m$-morphism $(M,M_M)$ is given by pushing $M$ forward along the constructible map
$$ \vec x \mapsto \begin{cases} \vec x& \text{ if } x_1 + \dots + x_n \geq 0; \\
(x_1,\dots,x_{n-1},-x_1-\dots-x_{n-1}) & \text{ if }  x_1 + \dots + x_n \leq 0. \end{cases} $$

Let $M$ now by an $n$-morphism with source $S$ and target $T$.  Then $S$ and $T$ are $1$-algebras and $M = (_SM_T,1_M)$ is a pointed $S$-$T$-bimodule (subject to compatibility conditions coming from the common source and target of $S$ and $T$).  The composition $M(S_S) = S_S \otimes_S {_S M _T}$ is the module $M_T$ in given by forgetting the $S$-action.  The pointing $1_M$ then determines a (strong, and in particular lax) pointed $T$-module map $T_T \to M_T$ given by $1_T \mapsto 1_M$.  Thus $(M,1_M)$ ``is'' an $n$-morphism in $\cat{Alg}_0^{\mathrm{lax}}(\cat{Alg}_n^{\mathrm{lax}}(\cS))$.  The construction $M \mapsto (M,1_M)$ extends naturally to lax pointed bimodule morphisms, and completes the description of the splitting $\cat{Alg}_n^{\mathrm{lax}}(\cS) \to \cat{Alg}_0^{\mathrm{lax}}(\cat{Alg}_n^{\mathrm{lax}}(\cS))$.
\end{proof}

In summary, the interesting dualizability questions related to $k$-affine Heisenberg-picture topological quantum field theory are about $d$-dualizability in $\cat{Alg}_n(\cS)$ for $\cS$ a $k$-fold delooping of $\cat{Vect}$ and $n+1 \leq d \leq n+k$.  

\section{From factorization algebra to Heisenberg-picture field theory} \label{factorization algebra example}

I would like now to explain an example which is based on unpublished work of Dwyer, Stolz, and Teichner.  I will outline my interpretation of, and elaboration upon, some parts of their construction, but details will need to wait for future work.  

\begin{definition}
Let $\cG$ be a not-necessarily-topological local geometry for $d$-dimensional manifolds, and $\cat{Emb}_d^\cG$ the symmetric monoidal  category of possibly-open $d$-dimensional $\cG$-geometric manifolds, with symmetric monoidal structure given by disjoint union.  Given a target category $\cS$, a \define{$\cG$-geometric prefactorization algebra valued in $\cS$} is a symmetric monoidal functor $F: \cat{Emb}_d^\cG \to \cS$.  In particular, for every $M \in \cat{Emb}_d^\cG$, $F(M)$ is pointed by $F(\emptyset \hookrightarrow M) : \unit_\cS \to F(M)$, and each embedding $M_1 \sqcup M_2 \hookrightarrow M$ determines a ``multiplication'' map $F(M_1) \otimes F(M_2) \to F(M)$.  A prefactorization algebra is a \define{factorization algebra} if it is local for the Weiss topology (see \cite{Ginot,CG,ClaudiaThesis} for details, but note that this usage of the phrase ``(pre)factorization algebra'' is not quite the same as the one in \cite{CG}).
\end{definition}

We will need two variations of the notion of ``factorization algebra.''  Let $X$ be a topological space.  A \define{factorization algebra on  $X$} is a precosheaf $F$ on $X$, local for the Weiss topology, such that if $U$ and $V$ are disjoint opens, then $F(U\sqcup V) \cong F(U) \otimes F(V)$.  Let $(X,\ast \in X)$ be a pointed topological space.  An \define{unpointed point-module on $(X,\ast\in X)$} is similar, but whereas in a factorization algebra there are maps $F(U) \to F(V)$ for every inclusion $U \subseteq V$ of open sets in $X$, in an unpointed point-module we do not ask for such maps in the special case where $U\ni \ast$ but $V \not\ni\ast$.  These are called ``unpointed'' because in the special case when $X = \{\ast\}$, an unpointed point-module is just an object of $\cS$, whereas a factorization algebra is a pointed object.
The restriction of a $\cG$-geometric factorization algebra to a space with $\cG$-geometry is a factorization algebra on that space; given a factorization algebra on a pointed space, one can forget to an unpointed point-module.

When $\cG$ is a not-necessarily-topological ``rigid supergeometry,'' a construction of the unextended spacetime category $\cat{Spacetimes}_{d-1,d}^\cG$ is described in \cite{MR2742432}.  In that construction, the objects are $(d-1)$-dimensional closed manifolds $N$ equipped with germs of one-sided collars diffeomorphic to $N \times [0,\epsilon)$ with $\cG$-structures on $N \times (0,\epsilon)$.  A morphism from $N_2$ to $N_1$ is a $d$-dimensional manifold $M$ with boundary $\partial M \cong N_1 \sqcup N_2$ and a germ of an extension of $M$ past $N_2$:
$$ 
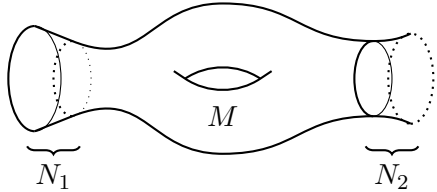
\begin{tikzpicture}
  \draw[thick] (-.5,.7) arc (90:270:.35 and .7);
  \draw[thin] (-.5,.7) arc (90:-90:.35 and .7);
  \draw[thick,dotted] (0,.5) arc (90:270:.25 and .5);
  \draw[thin,dotted] (0,.5) arc (90:-90:.25 and .5);
  \draw[thick] (4,.5) arc (90:270:.25 and .5);
  \draw[thin] (4,.5) arc (90:-90:.25 and .5);
  \draw[thick,dotted] (4.5,-.6) arc (-90:270:.3 and .6);
  \draw[thick] (-.5,.7) -- (0,.5) .. controls +(1,-.4) and +(-1,0) .. (2,1) .. controls +(1,0) and +(-1,0) .. (4,.5) .. controls +(.1,0) and +(-.2,-.1) .. (4.5,.6);
  \draw[thick] (-.5,-.7) -- (0,-.5) .. controls +(1,.4) and +(-1,0) .. (2,-1) .. controls +(1,0) and +(-1,0) .. (4,-.5) .. controls +(.1,0) and +(-.2,.1) .. (4.5,-.6);
  \draw[thick] (1.5,0) .. controls +(.3,.2) and +(-.3,.2) .. (2.5,0);
  \draw[thick] (1.35,.1) -- (1.5,0) .. controls +(.3,-.2) and +(-.3,-.2) .. (2.5,0) -- (2.65,.1);
  \draw[thick,decoration={brace,amplitude=3},decorate] (.1,-.9) -- coordinate (N1) (-.6,-.9) (N1) +(0,-.1) node[anchor=north]  {$N_1$};
  \draw[thick,decoration={brace,amplitude=3},decorate] (4.6,-.9) -- coordinate (N2) (3.9,-.9) (N2) +(0,-.1) node[anchor=north]  {$N_2$};
  \node at (2,-.5) {$M$};
\end{tikzpicture} $$

Given a $\cG$-geometric factorization algebra $F$, we will build a Heisenberg-picture field theory $\cZ_F:\cat{Spacetimes}_{d-1,d}^\cG \to \cat{Alg}_0(\cat{Cocomp}_\bK)$ for this version of $\cat{Spacetimes}_{d-1,d}^\cG$.

\begin{definitionnodiamond}
  Fix a $\cG$-geometric factorization algebra $F : \cat{Emb}_d^\cG \to \cat{Vect}_\bK$.
  Let $M$ be a (possibly open) $\cG$-manifold with boundary.  Let $M / \partial M$ denote the quotient (in topological spaces) in which $\partial M$ is contracted to a point; it is naturally pointed by that point.  The category $\cat{BC}_F(M)$ of \define{boundary conditions} for $F$ on $M$ is 
  $$ \cat{BC}_F(M) = \bigl\{\text{unpointed point-modules $\tilde F$ on $M / \partial M$ s.t.\ $\tilde F|_{M \sminus \partial M} = F|_{M \sminus \partial M}$}\bigr\}. $$
  The \define{free boundary condition} is the boundary condition given by pushing forward of $F|_{M \sminus \partial M}$ along $M \sminus \partial M \to M / \partial M$.
\end{definitionnodiamond}

\begin{lemma}\label{lemma.germiness of boundary conditions}
  The category of boundary conditions for $F$ on $M$ depends only on the germ of a $\cG$-manifold around $\partial M$.
\end{lemma}
\begin{proof}
  This follows from  descent/locality of factorization algebras \cite{ClaudiaThesis}.  Indeed, suppose that 
  $M' \hookrightarrow M$ is a submanifold with boundary such that the inclusion maps $\partial M' \isom \partial M$.  One can restrict unpointed point-modules on $M/\partial M$ to unpointed point-modules on $M' / \partial M'$.  Conversely, suppose we are given an unpointed point-module $\tilde F$ on $M' / \partial M'$ whose restriction to $M' \sminus \partial M'$ is $F|_{M' \sminus \partial M'}$.  Then it can be canonically extended to an unpointed point-module on $M / \partial M$ by gluing with $F |_{M \sminus \partial M}$.
\end{proof}

\begin{proposition}\label{prop.description of FN}
  Let $M$ be a manifold with boundary.  Its category $\cat{BC}_F(M)$ of boundary conditions is cocomplete and $\bK$-linear.
\end{proposition}

The proof will show indeed that $\cat{BC}_F(M)$ is locally presentable.

\begin{proof}
  I will give an alternate characterization of $\cat{BC}_F(M)$, closer to the construction of Dwyer, Stolz, and Teichner.  
  Consider the poset $\cP(M)$ whose objects are open neighborhoods of $\partial M$, ordered by inclusion.  We build a linear category $\cC_F(M)$ whose objects are the elements of $\cP(M)$ and whose morphisms are:
  $$ \hom(A,B) = \begin{cases}
    0, & \text{if } A \not \subseteq B \\
    F(B \sminus A), & \text{if } A \subseteq B. \\
  \end{cases} $$
  Note in particular that $\hom(A,A) = F(\emptyset) = \bK$.  The composition is given by the factorization algebra structure maps $F(A \sminus B) \otimes F(B \sminus C) \cong F((A \sminus B) \cup (B \sminus C)) \to F(A \sminus C)$.
  
  Let $\cat{Fun}_\bK(\cA,\cB)$ denote the category of $\bK$-linear functors between $\bK$-linear categories $\cA$ and $\cB$.
  There is a fully faithful embedding $\cat{BC}_F(M) \hookrightarrow \cat{Fun}_\bK(\cC_F(M),\cat{Vect}_\bK)$ that sends each unpointed point-module $\tilde F$ to the functor
  $ A \mapsto \tilde F(A \sminus \partial M). $  Functoriality is given by the factorization algebra structure.
  Since $\cat{Fun}_\bK(\cC_F(M),\cat{Vect}_\bK)$ is $\bK$-linear, so is $\cat{BC}_F(M)$.
  
  The essential image of this embedding consists of those $\bK$-linear functors $\cC(N) \to \cat{Vect}_\bK$ satisfying an appropriate ``locality'' condition coming from the locality condition for factorization algebras.  This locality condition consists of a series of assertions of the following form: $\tilde F(A)$ arises as a weighted colimit of a diagram each term of which is $\tilde F(B)$ for some $B \subsetneq A$, with the weightings being tensor products of $F(A \sminus B)$s.
  
  Thus $\cat{BC}_F(M)$ is the category of models of a colimit sketch structure on $\cC_F(M)$, and is therefore locally presentable and in particular cocomplete.
\end{proof}

With \cref{lemma.germiness of boundary conditions,prop.description of FN} in hand, we may define:

\begin{definition}
  Given an object $N \in \cat{Spacetimes}_{d-1,d}^\cG$ and a factorization algebra $F : \cat{Emb}_d^\cG \to \cat{Vect}_\bK$, we set $\cZ_F(N) = \cat{BC}_F(N \times[0,\epsilon)) \in \cat{Alg}_0^{\mathrm{lax}}(\cat{Cocomp}_\bK)$, where the pointing is via the free boundary condition.
\end{definition}

\begin{lemma}\label{lemma.symmetric monoidality of ZF}
  The assignment $N \mapsto \cZ_F(N)$ is symmetric monoidal in the sense that there are canonical equivalences $\cZ_F(\emptyset) = \cat{Vect}_\bK$ and $\cZ_F(N_1 \sqcup N_2) = \cZ_F(N_1) \boxtimes \cZ_F(N_2)$.
\end{lemma}

\begin{proof}
  The category of boundary conditions for the empty set is the category of unpointed point-modules on $\pt$, which is easily seen to be $\cat{Vect}_\bK$.  Assume $N_1$ and $N_2$ are nonempty.  Given $\bK$-linear categories $\cA$ and $\cB$, define their \define{naive tensor product} $\cA \otimes_\bK \cB$ to be the category with object set $\operatorname{ob}(\cA) \times \operatorname{ob}(\cB)$ and morphisms $\hom((A_1,B_1),(A_2,B_2)) = \hom_\cA(A_1,A_2) \otimes \hom_\cB(B_1,B_2)$.  Then there is a natural equivalence $\cat{Fun}_\bK(\cA,\cat{Vect}_\bK) \boxtimes \cat{Fun}_\bK(\cB,\cat{Vect}_\bK) = \cat{Fun}_\bK(\cA \otimes_\bK \cB,\cat{Vect}_\bK)$.  From this and the description of boundary conditions given in the proof of \cref{prop.description of FN} one may derive the natural equivalence $\cZ_F(N_1 \sqcup N_2) = \cZ_F(N_1) \boxtimes \cZ_F(N_2)$.
\end{proof}

\begin{definition}\label{defn.ZF on morphisms}
We now extend $\cZ_F$ to morphisms in $\cat{Spacetimes}_{d-1,d}^\cG$.  Let $M : N_2 \to N_1$ be a $1$-morphism, where we have chosen representative collars $N_1 \times[0,\epsilon_1)$ around $N_1$ and $N_2 \times[0,\epsilon_2)$ around $N_2$.  Since Heisenberg-picture field theories are assumed to be contravariantly functorial, we want to build a functor $\cZ_F(M) : \cZ_F(N_1) \to \cZ_F(N_2)$.

  By \cref{lemma.germiness of boundary conditions}, $\cZ_F(N_1) = \cat{BC}_F(M \sqcup_{N_2} N_2 \times [0,\epsilon_2))$, as the inclusion $N_1 \times [0,\epsilon_1) \to M \sqcup_{N_2} N_2 \times [0,\epsilon_2)$ is an isomorphism near the boundary.
Given a boundary condition $\tilde F \in \cat{BC}_F(M \sqcup_{N_2} N_2 \times [0,\epsilon_2))$, we may push it forward along the map $(M \sqcup_{N_2} N_2 \times [0,\epsilon_2)) / N_1 \to (N_2 \times [0,\epsilon_2))/N_2$ that is the identity on $N_2 \times (0,\epsilon)$ and collapses $M$ to the point.  The resulting pushed-forward factorization algebra then restricts over $N_2 \times (0,\epsilon)$ to $F|_{N_2 \times (0,\epsilon)}$ and so defines a boundary condition on $N_2 \times [0,\epsilon_2)$.
We set $\cZ_F(M) : \cZ_F(N_1) \to \cZ_F(N_2)$ to be this push-forward operation.

 Cocontinuity of $\cZ_F(M)$  follows from \cref{DST description of ZF}.  Its  structure as a lax pointed functor comes from the canonical maps $F(A_2 \sminus N_2) \to F((M\sminus N_1) \sqcup_{N_2} A_2)$, where $A_2 \subseteq N_2 \times [0,\epsilon)$ is an open neighborhood of the boundary $N_2$.  This defines the functor $\cZ_F : \cat{Spacetimes}_{d-1,d}^\cG \to \cat{Alg}_0^{\mathrm{lax}}(\cat{Cocomp}_\bK)$; functoriality follows from the functoriality of push-forward of factorization algebras, and symmetric monoidality can be proven by extending \cref{lemma.symmetric monoidality of ZF}.
\end{definition}

\begin{remark}\label{DST description of ZF}
  One can give $\cZ_F(N)$ a much more explicit description, completing the comparison to the construction described by Dwyer, Stolz, and Teichner.  I will continue with the notation from the proof of \cref{prop.description of FN}.
  
  Set $M = N \times [0,\epsilon)$.  For each object $A \in \cC_F(M)$, consider the full subcategory $\cC_F(M)_{\subsetneq A}$ of $\cC_F(M)$ on the objects $B$ with $B \subsetneq A$.  There are restriction functors $\cat{Fun}_\bK(\cC_F(M)_{\subsetneq A},\cat{Vect}_\bK) \to \cat{Fun}_\bK(\cC_F(M)_{\subsetneq B},\cat{Vect}_\bK)$ for $B \subsetneq A$ whose left adjoints define \define{extension} functors $$\cat{Fun}_\bK(\cC_F(M)_{\subsetneq B},\cat{Vect}_\bK) \to \cat{Fun}_\bK(\cC_F(M)_{\subsetneq A},\cat{Vect}_\bK).$$
  Unpacking the locality condition for factorization algebras reveals 
  $$
    \cZ_F(N) = \cat{BC}_F(M) =  \lim_{A \to \partial M} \cat{Fun}_\bK(\cC_F(M)_{\subsetneq A},\cat{Vect}_\bK),
  $$
  where the limit is taken over $\cP(M)$ along the extension functors.  It is worth noting that the colimit along restriction functors vanishes because $\hom(A,B) = 0$ if $A \not \subseteq B$.

  Let $M$ now be a $\cG$-geometric cobordism between $N_1$ and $N_2$.  There is a canonical inclusion $\cC_F(N_2\times [0,\epsilon_2)) \hookrightarrow \cC_F(M)$ given by $B \mapsto M \sqcup_{N_2} B$, and hence a restriction functor 
  $$ \cat{Fun}_\bK(\cC_F(M\sqcup_{N_2} N_2 \times [0,\epsilon_2))_{\subsetneq M \sqcup_{N_2} B},\cat{Vect}_\bK) \to \cat{Fun}_\bK(\cC_F(N_2 \times [0,\epsilon_2))_{\subsetneq B},\cat{Vect}_\bK). $$
  Composing this with the projection $\cZ_F(N_1) = \cat{BC}_F(M\sqcup_{N_2} N_2 \times [0,\epsilon_2)) \to \cat{Fun}_\bK(\cC_F(M\sqcup_{N_2} N_2 \times [0,\epsilon_2))_{\subsetneq M \sqcup_{N_2} B},\cat{Vect}_\bK)$ gives a sequence of cocontinuous functors
  $$ \cZ_F(N_1) \to \cat{Fun}_\bK(\cC_F(N_2 \times [0,\epsilon_2))_{\subsetneq B},\cat{Vect}_\bK).$$
  One may directly check that these commute with the extension functors for $B \hookrightarrow B'$, and so define a cocontinuous functor $\cZ_F(N_2) \to \cZ_F(N_2)$, which is nothing but the functor $\cZ_F(M)$ from \cref{defn.ZF on morphisms}.
  
  As a final remark, note that limits can be computed by passing to cofinal subcategories, and so $\cZ_F(N)$ can be presented as a limit over categories just for those opens that contract onto $N$.  Moreover, consider pushing forward $F|_{N \times (0,\epsilon)}$ along the projection $N \times (0,\epsilon) \to (0,\epsilon)$ to produce a factorization algebra on $(0,\epsilon)$.  One can then consider the category of boundary conditions on $[0,\epsilon)$ for this factorization algebra.  
By comparing descriptions in terms of limits of functor categories along extension functors, one can show that this category is nothing but $\cZ_F(N)$.  With this description and some careful study of limits of locally presentable categories, one can show in particular that when $F$ is locally constant, $\cZ_F(N) = \cat{Mod}_{\int_N F}$, and so the construction above matches the factorization homology functor from \cite{ClaudiaThesis}.
\end{remark}

  The construction of $\cZ_F(N)$ did not require that $N$ was compact.  Suppose that $N$ is an open $(n-1)$-dimensional manifold along with a germ of a $\cG$-structure on $ N \times [0,\epsilon)$; then there is still a category $\cat{BC}_F(N\times [0,\epsilon))$ of boundary conditions for $F|_{N\times (0,\epsilon)}$, and it is still described as in the proof of \cref{prop.description of FN}. 
  
  \begin{theorem}
    The assignment $N \mapsto \cZ_F(N)$ defines a symmetric monoidal precosheaf (i.e.\ a prefactorization algebra) valued in $\cat{Cocomp}_\bK$ on the category $\cat{Emb}_{d-1}^\cG$ of $(d-1)$-dimensional manifolds equipped with germs of $\cG$-structures.
  \end{theorem}

In fact, $\cZ_F$ is a factorization algebra, but the proof will await  future work.

\begin{proof}
   Suppose that $N_1 \hookrightarrow N_2$ is an inclusion compatible with germs of $\cG$-structures.  The corresponding map $N_1 \times [0,\epsilon) \hookrightarrow N_2 \times [0,\epsilon)$ extends to an inclusion
   $$ \bigl(N_1 \times [0,\epsilon)\bigr) \underset{N_1 \times (0,\epsilon)}\sqcup \bigl( N_2 \times (0,\epsilon) \bigr) \hookrightarrow \bigl( N_2 \times [0,\epsilon)\bigr), $$
   which in turn descends to a continuous function
   $$ \bigl(N_1 \times [0,\epsilon) / N_1\bigr) \underset{N_1 \times (0,\epsilon)}\sqcup \bigl( N_2 \times (0,\epsilon) \bigr) \hookrightarrow \bigl( N_2 \times [0,\epsilon) / N_2\bigr) $$
   which is a bijection on points, although not in general a homeomorphism.
   
   Any boundary condition on $N_1 \times [0,\epsilon)$ extends canonically to an unpointed point-module on $\bigl(N_1 \times [0,\epsilon) / N_1\bigr) \sqcup_{N_1 \times (0,\epsilon)} \bigl( N_2 \times (0,\epsilon) \bigr)$ by gluing with $F|_{N_2 \times (0,\epsilon)}$.  Pushing forward along the bijection above defines the functor $\cZ_F(N_1) \to \cZ_F(N_2)$.
   
   Symmetric monoidality of $\cZ_F$ follows from \cref{lemma.symmetric monoidality of ZF}.
\end{proof}

Nowhere in the construction of $\cZ_F$ did we need that the factorization algebra $F$ took values in $\cat{Vect}_\bK$ --- any locally presentable target category would have sufficed.  
Let $\cat{Pres}_\bK$ denote the bicategory of $\bK$-linear locally presentable categories; it is a full subbicategory of $\cat{Cocomp}_\bK$.
 $\cat{Pres}_\bK$  seems in many ways like it is a ``locally presentable bicategory'': it is cocomplete \cite{BirdThesis}; every object has a presentation \cite[Theorem 1.46]{MR1294136}.
 
The theory of locally presentable higher categories is under active development; cf.\ \cite{MR3345192}.  Eventually, one should expect to be able to iterate the above construction $F \leadsto \cZ_F$ to produce an extended field theory $\cZ_{\cZ_F}$ valued in ``locally presentable $\cat{Pres}_\bK$-linear bicategories.''  Iterating further would produce a fully extended Heisenberg-picture field theory from a factorization algebra.

\section{From skein theory to  Heisenberg-picture field theory} \label{section.skein theory}

In this final example I would like to describe a family of topological Heisenberg-picture field theories which one can prove are not affine.  The construction is my interpretation of the Reshetikhin--Turaev invariants of knots and links \cite{MR1036112}.  The extension to a local field theory is inspired by Walker's work \cite{WalkerTQFTs,MR2806651}.  It is the three-dimensional extension of the two-dimensional quantum field theories studied in \cite{MR3847209,MR3874702} and is expected to match the answer given by \cite{MR3818096}.  
  Recall some standard definitions:

\begin{definitionnodiamond}
  Let $\cC = (\cC,\otimes,\dots)$ be a small $\bK$-linear  monoidal category.  (The ``$\dots$'' denote the auxiliary data of an associator, unit, and unitors that I will generally suppress.)     A \define{strong monoidal functor} $(F,f)$ consists of a functor $F$ and a natural isomorphism $f:F(-)\otimes F(-)\isom F(-\otimes -)$ expressing compatibility with the monoidal structure, which must itself be compatible with associators (and also an isomorphism expressing compatibility with the units, but that isomorphism can always be canonically strictified to an identity, and so I will suppress it from the notation).
  Let $\otimes^\op$ denote the opposite monoidal structure on $\cC$.
  
  The monoidal category $\cC$ is \define{braided} if it is equipped with a strong monoidal functor $(\id_\cC,\beta): (\cC,\otimes,\dots) \to (\cC,\otimes^{\op},\dots)$ whose underlying functor is the identity functor $\id_\cC$.  (This is equivalent to the usual hexagon relations for $\beta$.)
    A \define{full twist} for $\cC$ is is an isomorphism $\theta :  (\id_\cC,\beta^{-1}) \isom (\id_\cC,\beta)$ of monoidal functors between $(\cC,\otimes)$ and $(\cC,\otimes^{\op})$, or equivalently an isomorphism of monoidal endofunctors $\theta : (\id_{\cC},\id) \isom (\id_{\cC},\beta^2)$ of $(\cC,\otimes)$.  Spelled out, a full twist consists of a natural automorphism $\theta_X : X \isom X$ such that
  $$ \beta_{X,Y} \circ (\theta_X \otimes \theta_Y) =  \beta_{Y,X}^{-1} \circ \theta_{X\otimes Y}$$
  for all $X,Y \in \cC$.
  
    The name ``full twist'' comes from interpreting $\theta_{X}$ as a $360^{\circ}$ twist of a ribbon labeled by $X$:
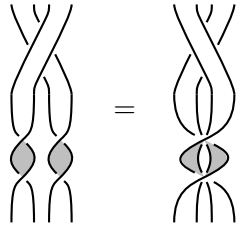
$$ 
\begin{tikzpicture}[thick,baseline=(midpoint)]
  \coordinate (midpoint) at (.4,1.4);
  \path[fill=black!25] (.15,.6) .. controls +(.2,.2) and +(.2,-.2) .. (.15,1.1) .. controls +(-.2,-.2) and +(-.2,.2) .. (.15,.6);
  \draw[shorten >=2pt] (0,0) .. controls +(0,.5) and +(-.1,-.1) .. (.15,.6) .. controls +(.2,.2) and +(.2,-.2) .. (.15,1.1);
  \draw[shorten >=2pt] (.3,0) .. controls +(0,.5) and +(.1,-.1) .. (.15,.6) ;
  \draw[shorten <=2pt] (.15,.6) .. controls +(-.2,.2) and +(-.2,-.2) .. (.15,1.1) .. controls +(.1,.1) and +(0,-.5) .. (.3,1.7);
  \draw[shorten <=2pt] (.15,1.1) .. controls +(-.1,.1) and +(0,-.5) ..  (0,1.7);
  \path[fill=black!25] (.65,.6) .. controls +(.2,.2) and +(.2,-.2) .. (.65,1.1) .. controls +(-.2,-.2) and +(-.2,.2) .. (.65,.6);
  \draw[shorten >=2pt] (.5,0) .. controls +(0,.5) and +(-.1,-.1) .. (.65,.6) .. controls +(.2,.2) and +(.2,-.2) .. (.65,1.1);
  \draw[shorten >=2pt] (.8,0)  .. controls +(0,.5) and +(.1,-.1) .. (.65,.6) ;
  \draw[shorten <=2pt] (.65,.6) .. controls +(-.2,.2) and +(-.2,-.2) .. (.65,1.1) .. controls +(.1,.1) and +(0,-.5) .. (.8,1.7);
  \draw[shorten <=2pt] (.65,1.1) .. controls +(-.1,.1) and +(0,-.5) ..  (.5,1.7);
  \draw (.3,1.7) .. controls +(0,.1) and +(-.1,-.2) .. (.4,2) --  (.55,2.3) .. controls +(.1,.2) and +(0,-.1) .. (.8,2.9);
  \draw (0,1.7) .. controls +(0,.1) and +(-.1,-.2) .. (.25,2.3) -- (.4,2.6) .. controls +(.1,.2) and +(0,-.1) .. (.5,2.9);
  \draw[shorten >=2pt] (.5,1.7) .. controls +(0,.1) and +(.1,-.2) .. (.4,2) ;
  \draw[shorten <=2pt] (.25,2.3) .. controls +(-.1,.2) and +(0,-.1) .. (0,2.9);
  \draw[shorten >=2pt] (.8,1.7) .. controls +(0,.1) and +(.1,-.2) .. (.55,2.3);
  \draw[shorten <=2pt] (.4,2.6) .. controls +(-.1,.2) and +(0,-.1) .. (.3,2.9);
\end{tikzpicture}
\quad = \quad
\begin{tikzpicture}[thick,baseline=(midpoint)]
  \coordinate (midpoint) at (.4,1.4);
  \path[fill=black!25] (.4,.6) .. controls +(.4,.2) and +(.4,-.2) .. (.4,1.1) .. controls  +(.1,-.2) and +(.1,.2) .. (.4,.6);
  \path[fill=black!25] (.4,.6) .. controls +(-.4,.2) and +(-.4,-.2) .. (.4,1.1) .. controls  +(-.1,-.2) and +(-.1,.2) .. (.4,.6);
  \draw[shorten >=4pt] (0,0) .. controls +(0,.4) and +(-.2,-.1) .. (.4,.6) .. controls +(.4,.2) and +(.4,-.2) .. (.4,1.1);
  \draw[shorten >=2pt] (.3,0) .. controls +(0,.2) and +(-.1,-.2) .. (.4,.6);
  \draw[shorten <=2pt,shorten >=2pt] (.4,.6) .. controls +(.1,.2) and +(.1,-.2) .. (.4,1.1);
  \draw[shorten >=2pt] (.5,0) .. controls +(0,.2) and +(.1,-.2) .. (.4,.6);
  \draw[shorten >=4pt] (.8,0) .. controls +(0,.4) and +(.2,-.1) .. (.4,.6);
  \draw[shorten <=2pt,shorten >=2pt] (.4,.6) .. controls +(-.1,.2) and +(-.1,-.2) .. (.4,1.1);
  \draw[shorten <=4pt] (.4,.6) .. controls +(-.4,.2) and +(-.4,-.2) .. (.4,1.1) .. controls +(.2,.1) and +(0,-.4) .. (.8,1.7);
  \draw[shorten <=2pt] (.4,1.1) .. controls +(.1,.2) and +(0,-.2) .. (.5,1.7);
  \draw[shorten <=2pt] (.4,1.1) .. controls +(-.1,.2) and +(0,-.2) .. (.3,1.7);
  \draw[shorten <=4pt] (.4,1.1) .. controls +(-.2,.1) and +(0,-.4) .. (0,1.7);
  \draw (.5,1.7) .. controls +(0,.1) and +(.1,-.2) .. (.4,2) --  (.25,2.3) .. controls +(-.1,.2) and +(0,-.1) .. (0,2.9);
  \draw (.8,1.7) .. controls +(0,.1) and +(.1,-.2) .. (.55,2.3) -- (.4,2.6) .. controls +(-.1,.2) and +(0,-.1) .. (.3,2.9);
  \draw[shorten >=2pt] (.3,1.7) .. controls +(0,.1) and +(-.1,-.2) .. (.4,2) ;
  \draw[shorten <=2pt] (.55,2.3) .. controls +(.1,.2) and +(0,-.1) .. (.8,2.9);
  \draw[shorten >=2pt] (0,1.7) .. controls +(0,.1) and +(-.1,-.2) .. (.25,2.3);
  \draw[shorten <=2pt] (.4,2.6) .. controls +(.1,.2) and +(0,-.1) .. (.5,2.9);
\end{tikzpicture} $$
Categories equipped with a full twist are  called \define{balanced} in \cite{MR1113284}, but I will not use this word so as not to conflict with the notion of ``balanced tensor product'' from \cref{alg1cocomp,defn.naive balanced tensor product}.

A small monoidal category $\cC$ \define{has duals} if for every object $X \in \cC$ the functor $X\otimes : \cC \to \cC$ has 
a right adjoint of the form ${^*X} \otimes$ and a left adjoint of the form ${X^*} \otimes$ for some objects ${X^*},{^*X} \in \cC$.  (These objects are unique up to unique isomorphism if they exist, in which case $\otimes {X^*}$ and $\otimes {^*X}$ are the left and right adjoints respectively to $\otimes X$.  Compare \cref{defn.dualizable}.)  If $\cC$ has duals, then there is a canonically defined \define{double dual} functor $X \mapsto X^{**}$ which is a monoidal equivalence $(\cC,\otimes,\dots) \to (\cC,\otimes,\dots)$.
Suppose that $\cC = (\cC,\otimes,\beta,\dots)$ is a small $\bK$-linear braided monoidal category with duals.  The braiding determines  isomorphisms 
$$ \tau_{\beta} : (-)^{**}\isom (\id,\beta^{2}) \quad \text{and} \quad \tau_{\beta^{-1}} : (-)^{**} \isom (\id,\beta^{-2}) $$
of monoidal functors $\cC \to \cC$, via
$$
\begin{tikzpicture}[baseline=(A.base),anchor=base,auto]
  \path node (A) {$X^{**}$} ++(5,0) node (B) {$X^{**} \otimes X \otimes X^*$} ++(5,0) node (C) {$X \otimes X^{**} \otimes X^*$} ++(5,0) node (D) {$X$} (B.east) +(0,2pt) coordinate (raised) +(0,-2pt) coordinate (lowered);
  \draw[->] (A) -- node {\footnotesize unit of adjunction} (B);
  \draw[->] (C) -- node {\footnotesize counit of adjunction} (D);
  \draw[->] (B.east |- raised) -- node {\footnotesize $\beta_{X^{**},X} \otimes \id_{X^*}$} (C.west |- raised);
  \draw[->] (B.east |- lowered) -- node[swap] {\footnotesize $\beta^{-1}_{X,X^{**}} \otimes \id_{X^*}$} (C.west |- lowered);
  \draw[->] (A) .. controls +(1,1) and +(-1,1) .. node {\footnotesize $\tau_\beta(X)$} (D);
  \draw[->] (A) .. controls +(1,-1) and +(-1,-1) .. node[swap] {\footnotesize $\tau_{\beta^{-1}}(X)$} (D);
\end{tikzpicture}
$$
with inverses
$$
\begin{tikzpicture}[baseline=(A.base),anchor=base,auto]
  \path node (A) {$X$} ++(5,0) node (B) {$X^{*} \otimes X^{**} \otimes X$} ++(5,0) node (C) {$X^{*} \otimes X \otimes X^{**}$} ++(5,0) node (D) {$X^{**}$} (B.east) +(0,2pt) coordinate (raised) +(0,-2pt) coordinate (lowered);
  \draw[->] (A) -- node {\footnotesize unit of adjunction} (B);
  \draw[->] (C) -- node {\footnotesize counit of adjunction} (D);
  \draw[->] (B.east |- raised) -- node {\footnotesize $\id_{X^*} \otimes \beta^{-1}_{X,X^{**}}$} (C.west |- raised);
  \draw[->] (B.east |- lowered) -- node[swap] {\footnotesize $\id_{X^*} \otimes \beta_{X^{**},X}$} (C.west |- lowered);
  \draw[->] (A) .. controls +(1,1) and +(-1,1) .. node {\footnotesize $\tau_\beta^{-1}(X)$} (D);
  \draw[->] (A) .. controls +(1,-1) and +(-1,-1) .. node[swap] {\footnotesize $\tau_{\beta^{-1}}^{-1}(X)$} (D);
\end{tikzpicture}
$$

A braided monoidal category $\cC$ with duals is \define{ribbon} if it is equipped with a full twist $\theta$ satisfying the quadratic equation $\theta^2 = \tau_{\beta^{-1}}^{-1}\circ \tau_\beta$:
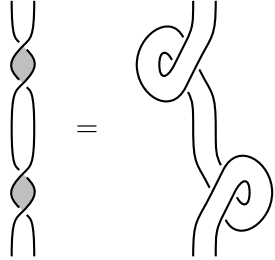
$$ 
\begin{tikzpicture}[thick,baseline=(midpoint)]
  \coordinate (midpoint) at (.4,1.6);
  \path[fill=black!25] (.15,.6) .. controls +(.2,.2) and +(.2,-.2) .. (.15,1.1) .. controls +(-.2,-.2) and +(-.2,.2) .. (.15,.6);
  \path[fill=black!25] (.15,2.3) .. controls +(.2,.2) and +(.2,-.2) .. (.15,2.8) .. controls +(-.2,-.2) and +(-.2,.2) .. (.15,2.3);
  \draw[shorten >=2pt] (0,0) .. controls +(0,.5) and +(-.1,-.1) .. (.15,.6) .. controls +(.2,.2) and +(.2,-.2) .. (.15,1.1);
  \draw[shorten >=2pt] (.3,0) .. controls +(0,.5) and +(.1,-.1) .. (.15,.6) ;
  \draw[shorten <=2pt,shorten >=2pt] (.15,.6) .. controls +(-.2,.2) and +(-.2,-.2) .. (.15,1.1) .. controls +(.1,.1) and +(0,-.5) .. (.3,1.7) .. controls +(0,.5) and +(.1,-.1) .. (.15,2.3) ;
  \draw[shorten <=2pt,shorten >=2pt] (.15,1.1) .. controls +(-.1,.1) and +(0,-.5) ..  (0,1.7) .. controls +(0,.5) and +(-.1,-.1) .. (.15,2.3) .. controls +(.2,.2) and +(.2,-.2) .. (.15,2.8) ;
  \draw[shorten <=2pt] (.15,2.3) .. controls +(-.2,.2) and +(-.2,-.2) .. (.15,2.8) .. controls +(.1,.1) and +(0,-.5) .. (.3,3.4);
  \draw[shorten <=2pt]  (.15,2.8) .. controls +(-.1,.1) and +(0,-.5) .. (0,3.4);
\end{tikzpicture}
\quad = \quad
\begin{tikzpicture}[thick,baseline=(midpoint)]
  \coordinate (midpoint) at (.4,1.6);
  \draw[shorten >=2pt] (.3,0) .. controls +(0,.3) and +(-.1,-.2) .. (.4,.55) -- (.55,.85) .. controls +(.1,.2) and +(0,.2) .. (.75,.85) .. controls +(0,-.2) and +(.1,-.2) .. (.55,.85);
  \draw[shorten >=2pt] (0,0) .. controls +(0,.5) and +(-.2,-.4) .. (.25,.85) -- (.4,1.15) .. controls +(.2,.4) and +(0,.4) .. (1.05,.85) .. controls +(0,-.4) and +(.2,-.4) .. (.4,.55);
  \draw[shorten <=2pt,shorten >=2pt] (.25,.85) .. controls +(-.2,.4) and +(0,-.5) .. (0,1.7) .. controls +(0,.3) and +(.1,-.2) .. (-.1,2.25);
  \draw[shorten <=2pt,shorten >=2pt] (.4,1.15) .. controls +(-.1,.2) and +(0,-.3) .. (.3,1.7) .. controls +(0,.5) and +(.2,-.4) .. (.05,2.55);
  \draw[shorten <=2pt] (-.25,2.55) .. controls +(-.1,.2) and +(0,.2) .. (-.45,2.55) .. controls +(0,-.2) and +(-.1,-.2) .. (-.25,2.55) -- (-.1,2.85) .. controls +(.1,.2) and +(0,-.3) .. (0,3.4);
  \draw[shorten <=2pt] (-.1,2.85) .. controls +(-.2,.4) and +(0,.4) .. (-.75,2.55) .. controls +(0,-.4) and +(-.2,-.4) .. (-.1,2.25) -- (.05,2.55) .. controls +(.2,.4) and +(0,-.5) .. (.3,3.4);
\end{tikzpicture}
$$

Let $(\cC,\otimes,\beta,\theta,\dots)$ be a ribbon category.  A \define{$\cC$-labeled ribbon tangle in $[0,1]^3$} is a piecewise-smooth embedding (except for as described in item (\ref{item.embedding rules})) into $[0,1]^3$ of a finite collection of rectangles subject to the following conditions:
  \begin{enumerate}
    \item The rectangles come in two kinds, ``long thin'' \define{ribbons} and ``short squat'' \define{coupons}.  The \define{bottom end} of a ribbon is the subset $[0,1] \times \{0\}$ of its boundary, and the \define{top end} if $[0,1]\times \{1\}$.  The \define{bottom side} of a coupon is the subset $[0,1] \times \{0\}$ of its boundary, and the \define{top side} is $[0,1] \times \{1\}$.
    \item \label{item.embedding rules} Each bottom end of each ribbon is a subset either of the \define{bottom} of the cube $[0,1]^2 \times \{0\}$ or of a top  side of a coupon; in the latter case the map $[0,1] \times \{0\} \hookrightarrow [0,1] \times \{1\}$ mapping the end of the ribbon into the top or bottom of the coupon is orientation-preserving.  The top end of each ribbon is a subset either of the \define{top} of the cube $[0,1]^2 \times \{1\}$ or of a bottom side of a coupon; again in the latter case orientations should be preserved.  
    These are the only intersections of ribbons and coupons with each other or with the boundary of $[0,1]^3$.
    \item Each ribbon is labeled with an object of $\cC$.  Each coupon is labeled with a morphism in $\cC$ compatibly with the ribbons that end on it: if the ribbons ending on its bottom are labeled in order $X_1,\dots,X_m$ and the ribbons ending on its top are labeled in order $Y_1,\dots,Y_n$, then the coupon itself is labeled with an element of $\hom(X_1\otimes \dots \otimes X_m,Y_1\otimes\dots\otimes Y_n)$. \diamondhere
  \end{enumerate}
\end{definitionnodiamond}

The main coherence theorem about ribbon categories is the following:

\begin{theorem}[\cite{MR1036112}]\label{thm.rt2}
  There is a well-defined \define{interpretation} map from $\cC$-labeled ribbon tangles in $[0,1]^3$ to morphisms in $\cC$.  The source and target of the morphism depend only on the intersections of the tangle with the bottom and top, respectively, of $[0,1]^3$.  The interpretation of a tangle depends only on the isotopy-rel-boundary class of the tangle.
  
  Suppose that $[0,1]^3 \hookrightarrow [0,1]^3$ is an orientation-preserving nesting of a smaller cube into the interior of a larger cube.  Suppose further that we are given two $\cC$-labeled ribbon tangles in the larger cube such that (i) the two tangles agree outside the smaller cube, (ii) the tangles intersect  the boundary of smaller cube only along its top and bottom, and (iii) the intersections of the two tangles with the smaller cube have the same interpretation.  Then the two tangles have the same interpretation in the larger cube. \qedhere
\end{theorem}

\Cref{thm.rt2} allows us to use $\cC$ to study tangles in more complicated manifolds.  The following construction is particularly important:

\begin{definition} \label{defn.skein category}
  Let $\cC$ be a ribbon category and let $N$ be a possibly-open oriented surface admitting a finite cover by disks.  The \define{skein category} $\int_N \cC$ is the $\bK$-linear category described as follows:
  \begin{description} \item[objects] configurations of disjoint ``short'' oriented intervals in $N$, each labeled by an object of $\cC$. \end{description}
   There is a natural extension of the notion of ``ribbon tangle labeled by $\cC$'' to tangles in $N \times [0,1]$.  Note that for any ribbon tangle, its intersections with $N \times \{0\}$ and with $N \times \{1\}$ are naturally objects of $\int_N\cC$.
   \begin{description} \item[morphisms] the space of morphisms from $X_0$ to $X_1$ is spanned by the set of $\cC$-labeled ribbon tangles that intersect $N \times \{i\}$ at $X_i$, modulo the following relation: for every  orientation-preserving embedding $[0,1]^3 \hookrightarrow N \times[0,1]$ of a ``small'' cube into the interior of $N \times [0,1]$, we set  to zero any linear combination of tangles $\sum T_a$ if (i) all the $T_a$s are equal outside the small cube, (ii) they intersect the small cube only along its top and bottom sides, and the intersection is transverse, and (iii) $\sum \operatorname{interpret}(T_a \cap (\text{small cube})) = 0$.
   \end{description}
   It follows from \cref{thm.rt2} that isotopic-rel-boundary tangles represent the same morphism in $\int_N \cC$.  Composition is by stacking of tangles in $N \times [0,1] \cup_N N \times [0,1] \simeq N \times [0,1]$.
   The skein category $\int_N \cC$ is naturally pointed by the empty object.
   If $N$ has boundary, we set $\int_N \cC = \int_{\mathring N}\cC$.
\end{definition}

\begin{definition}\label{defn.naive balanced tensor product}
  Let $\cA$ be a small $\bK$-linear monoidal category with a right module $\cX$ and a left module $\cY$. 
  Recall from the proof of \cref{lemma.symmetric monoidality of ZF} that the \define{naive tensor product} of $\cX$ with $\cY$ is the category $\cX \otimes_\bK \cY$ with object set $\mathrm{ob}(\cX) \times \mathrm{ob}(\cY)$ and morphisms $\hom((X_1,Y_1),(X_2,Y_2)) = \hom_\cX(X_1,X_2) \otimes \hom_\cY(Y_1,Y_2)$.
   The \define{naive balanced tensor product} $\cX \otimes_\cA \cY$ is the $\bK$-linear category formed from $\cX \otimes_\bK \cY$ by adding, for each object $(X,A,Y) \in \cX \otimes_\bK \cA \otimes_\bK \cY$, a natural-in-$(X,A,Y)$ isomorphism $\alpha_{X,A,Y} : (X \triangleleft A)\otimes Y \isom X \otimes(A \triangleright Y)$, subject to a pentagon relation (and also adding some unitor relations that we will not write down).  Compare \cref{alg1cocomp}.
\end{definition}

The first part of the following result is straightforward. The excision property is tedious, and is worked out in full in \cite{Cooke19}. (That paper was based, in turn, on an earlier preprint of this paper.)
\begin{theorem}[\cite{Cooke19}]\label{Cooke theorem}
  The skein category $\int_{N_1 \sqcup N_2} \cC$ of a disjoint union is the {naive tensor product} of skein categories $(\int_{N_1}\cC) \otimes_\bK (\int_{N_2}\cC)$.
  
  The construction $\int_\Box\cC : \{\text{surfaces}\} \to \{\text{categories}\}$ is functorial on embeddings.  It follows that for any oriented one-manifold $P$, $\int_P \cC = \int_{P \times [0,1]} \cC$ is naturally a monoidal category and for any oriented zero-manifold $P$, $\int_P \cC = \int_{P \times [0,1]^2} \cC$ is naturally a braided monoidal category.  It follows also that if $N$ is a surface with boundary $\partial N$, then $\int_{ N} \cC$ is naturally a module category for $\int_{\partial N} \cC$. The interpretation map from \cref{thm.rt2} provides a braided monoidal equivalence $\int_{\pt} \cC \isom \cC$. 
  
  Moreover, the functor $\int_\Box\cC : \{\text{surfaces}\} \to \{\text{categories}\}$ satisfies the following version of \define{excision}.  Suppose that $N$ is a surface with a decomposition as $N = N_1 \sqcup_P N_2$ along a one-manifold $P$.  Then $\int_{ N_1} \cC$ is naturally a right $\int_{P}\cC$-module and $\int_{ N_2}\cC$ is naturally a left module, and $\int_N\cC$ is equivalent to the {naive balanced tensor product} 
  
  \hfill $\displaystyle \int_{ N} \cC \simeq \left(\int_{ N_1}\cC\right) \underset{\int_{P }\cC}\otimes \left(\int_{ N_2}\cC\right).$ \qedhere
\end{theorem}

Combining results from \cite{AF2012,ClaudiaThesis} gives:
\begin{corollary}\label{cor.cooke}
  The assignment $N \mapsto \int_{ N} \cC$ defines a symmetric monoidal functor $\cat{Bord}_{2}^{\mathrm{or}} \to \cat{Alg}_2^{\mathrm{lax}}(\cat{Cat}_\bK)$, where $\cat{Bord}_{2}^{\mathrm{or}}$ is the fully-extended oriented 2-bordism category and $\cat{Cat}_\bK$ is   
   the bicategory of small linear categories and linear functors. \qedhere
\end{corollary}

This is not the preferred target of a Heisenberg-picture field theory --- as in \cref{defn.1-affine extended qft}, I would much rather land in $\cat{Alg}_2^{\mathrm{lax}}(\cat{Cocomp}_\bK)$.  But this is not a problem.  The following is proved in \cite{MR2177301}:

\begin{lemma}
  Let $\cX$ be a small $\bK$-linear category.  The \define{free cocompletion} $\widehat \cX$ of $\cX$ is $\cat{Fun}_\bK(\cX^{\op},\cat{Vect}_\bK)$.  It satisfies the following universal property: for any $\cE \in \cat{Cocomp}_\bK$, there is a natural equivalence $\cat{Fun}_\bK(\cX,\cE) \simeq \hom_{\cat{Cocomp}_\bK}(\widehat\cX,\cE)$.  Free cocompletion is symmetric monoidal: $\widehat{\cX \otimes_\bK \cY} \simeq \widehat \cX \boxtimes_\bK \widehat \cY$.  The essential image of $\widehat{(-)}$ consists of those cocomplete categories generated under colimits by some set of compact projective objects.  There is a natural equivalence
  
  \hfill $\displaystyle\hom_{\cat{Cocomp}_\bK}(\widehat\cX,\widehat\cY) \simeq \cat{Fun}_\bK(\cX \otimes_\bK \cY^{\op},\cat{Vect}_\bK).$ \qedhere
\end{lemma}
Linear functors $\cX \otimes_\bK \cY^{\op} \to \cat{Vect}_\bK$ should be thought of as \define{bimodules} between the categories $\cX$ and $\cY$, where $\cX$ and $\cY$ themselves are ``many object associative algebras.''

\begin{corollary}
  Free cocompletion takes naive balanced tensor products to balanced tensor products:
  $$ \widehat{ \cX \otimes_\cA \cY} \simeq \widehat \cX \boxtimes_{\widehat \cA} \widehat \cY $$
  if $\cA$ is a small $\bK$-linear monoidal category with a left module $\cX$ and a right module $\cY$. 
\end{corollary}
\begin{proof}
Both sides satisfy the same universal property.
\end{proof}

It follows that $N \mapsto \widehat {\int_{ N} \cC}$ defines a symmetric monoidal functor $\cat{Bord}_{2}^{\mathrm{or}} \to \cat{Alg}_2^{\mathrm{lax}}(\cat{Cocomp}_\bK)$.  This is not yet a Heisenberg-picture field theory: it assigns pointed linear categories, not pointed vector spaces, to top-dimensional manifolds.  To build a Heisenberg-picture field theory requires extending this functor to three-manifolds.

\begin{definition}\label{defn.skeinmodule}
Let $M$ be a compact three-dimensional manifolds with boundary decomposed as $\partial M = N_{\mathrm{bot}} \sqcup_{P \times \{0\}} P \times [0,1] \sqcup_{P \times \{1\}} N_{\mathrm{top}}$, where $P$ is a one-manifold and $N_{\mathrm{bot}}$ and $N_{\mathrm{top}}$ are the ``bottom'' and ``top'' parts of the boundary.  Such manifolds are the three-morphisms in $\cat{Bord}_3^{\mathrm{or}}$.  
Given a ribbon category $\cC$,
the \define{skein module} $\int_M\cC$ of $M$ is (the span of) ribbon tangles in $M$ modulo local relations for little cubes $[0,1]^3 \hookrightarrow M$ just as in the morphisms of \cref{defn.skein category}.  By stacking on cylinders over $ N_{\mathrm{bot}}$ and $ N_{\mathrm{top}}$, $\int_M\cC$ is naturally an $\int_{ N_{\mathrm{bot}}}\cC$--$\int_{ N_{\mathrm{top}}}\cC$--bimodule, i.e.\ a functor 
$$ \left( \int_{ N_{\mathrm{bot}}}\cC \right) \otimes_\bK \left(\int_{ N_{\mathrm{top}}}\cC\right)^{\op} \to \cat{Vect}_\bK,$$
and so defines a cocontinuous functor
$$ \widehat{\int_M\cC} : \widehat {\int_{ N_{\mathrm{bot}}}\cC} \to \widehat{\int_{ N_{\mathrm{top}}}\cC}.$$
The empty tangle gives $\widehat{\int_M\cC}$ the structure of a lax pointed functor.
\end{definition}

\begin{conjecture}\label{skein theory thm}
When $P$ is nonempty, for each object $X \in \int_{P} \cC$, there is a natural operation that takes a tangle in $M$ and 
outputs its union with the identity tangle over $X$.
This operation makes $\widehat{\int_M\cC}$ into a lax module functor between the $\widehat{\int_{P}\cC}$-modules 
$\widehat{\int_{ N_{\mathrm{bot}}}\cC} $ and $ \widehat{\int_{ N_{\mathrm{top}}}\cC}$.  The assignment $\widehat{\int_\Box \cC}$ defines an oriented topological fully extended 1-affine Heisenberg-picture field theory $\cat{Bord}_3^{\mathrm{or}} \to \cat{Alg}_2^{\mathrm{lax}}(\cat{Cocomp}_\bK)$. 
\end{conjecture}

\begin{remark}
\Cref{skein theory thm} surely follows from the previous results, and is equivalent in spirit to results from \cite{WalkerTQFTs}. Nevertheless, I have labeled it a Conjecture because there are some subtle details that must be checked (and I invite the reader to convert such a check into a paper).  That $P \mapsto \int_P \cC$ defines a functor $\cat{Bord}_2^{\mathrm{or}} \to \cat{Alg}_2^{\mathrm{lax}}(\cat{Cocomp}_\bK)$ follows from cocompleting \Cref{cor.cooke}. 
This provides the values of the desired functor $\cat{Bord}_3^{\mathrm{or}} \to \cat{Alg}_2^{\mathrm{lax}}(\cat{Cocomp}_\bK)$ on $0$, $1$, and $2$-morphisms. The values on 3-morphisms are given by the skein modules of \cref{defn.skeinmodule}, and functoriality for 3-morphism composition is manifest. The subtle detail that needs checking is the compatibility between the composition of 3-morphisms and the composition of 2-morphisms.
\end{remark}

  A braided monoidal cocomplete category $\cA \in \cat{Alg}_2^{\mathrm{lax}}(\cat{Cocomp}_\bK)$ is of the form $\widehat \cC$ for $\cC$ a small ribbon category if and only if it satisfies all of the following:
  \begin{itemize}
    \item The unit object $\unit_\cA \in \cA$ is {compact projective}.
    \item Every object of $\cA$ is a colimit of dualizable objects.
    \item $\cA$ is equipped with a full twist whose restriction makes the full subcategory $\cA^{\mathrm{fd}}$ on the dualizable objects into a ribbon category.
    \item $\cA^{\mathrm{fd}}$ is equivalent to a small category.
  \end{itemize}
  For such $\cA$, $\cA = \widehat{\cA^{\mathrm{fd}}}$.  In general, $(\widehat\cC)^{\mathrm{fd}}$ is the \define{Karoubi envelop} or \define{Cauchy completion} of $\cC$: it's what you get when you take $\cC$, add all finite direct sums, and then split all idempotents.
  
  These conditions are satisfied, for example, for $\cA = \cat{Rep}^{\mathrm{integrable}}_{\bC(q)}(U_q\mathfrak g)$ the category of ind-finite-dimensional representations of a quantum group at generic quantum parameter $q$, or for $\cA = \cat{Rep}_\bC^{\mathrm{algebraic}}(G)$ the category of algebraic representations of a reductive group when $\bC$; the dualizable objects are then the finite-dimensional modules.  The conditions fail for (non-semisimplified) quantum groups at roots of unity and for  finite groups over fields of characteristic dividing the order of the group.

\begin{example}
Let's say that a \define{trivial} ribbon category is the category $\cat{Mod}^{\mathrm{fd}}_A$ of finitely generated projective modules over a commutative $\bK$-algebra $A$ with $\otimes = \otimes_A$, or some ribbon subcategory thereof.  
The monoidal category $(\cat{Mod}^{\mathrm{fd}}_A,\otimes_A)$ admits a unique braiding and unique full twist, since it is generated under direct sums and passing to direct summands by the monoidal unit.
All ribbon subcategories of $\cat{Mod}^{\mathrm{fd}}_A$ determine the same Heisenberg-picture field theory valued in $\cat{Cocomp}_\bK$, since they have canonically identified free cocompletions.  
The 1-affine Heisenberg-picture field theory $\widehat{\int_\Box\cC}$ is 0-affine if and only if $\cC$ is trivial.
\end{example}

\begin{example}
  If $G$ is a reductive group over $\bK = \bC$ and $N$ is a surface, then $\widehat{\int_N \cat{Rep}^{\mathrm{fd}}_\bC(G)} \simeq \cat{QCoh}(\mathrm{Loc}_G(N))$, where $\cat{QCoh}$ is the category of quasicoherent sheaves of $\cO$-modules and $\mathrm{Loc}_G(N)$ is the stack of $G$-local systems, presented by the groupoid $\hom(\pi_1(N),G)/G$ \cite{MR3847209,MR3874702}.  This can be proven by showing that $N \mapsto \cat{QCoh}(\hom(\pi_1(N),G)/G)$ satisfies excision.
  
  An object  $X = \sqcup \{X_i\} \in \int_N \cat{Rep}^{\mathrm{fd}}(G)$ represents the following sheaf on $\mathrm{Loc}_G(N)$: given a local system on $N$, restrict it at each of the small intervals $X_i : [0,1] \to N$ in $X$; tensor each restriction with the corresponding $G$-module to get a flat vector bundle on the $i$th interval; tensor together the spaces of flat sections of each of these bundles.
  
  Suppose that $M$ is a three-manifold with boundary $\partial M$, read as a morphism $M : \emptyset \to \partial M$.  Then $\widehat{\int_M \cC} : \cat{Vect}_\bK \to \widehat{\int_{\partial M}\cC} \simeq \cat{QCoh}(\mathrm{Loc}_G(\partial))$ is the pushforward of the structure sheaf $\cO(\mathrm{Loc}_G(M))$ along the restriction map $\mathrm{Loc}_G(M) \to \mathrm{Loc}_G(\partial M)$.
  \end{example}
  
\begin{example}
  Let $\cC = \cat{TL}$ be the well-known \define{Temperley--Lieb} category, which is a ribbon category defined over $\bK = \bZ[q,q^{-1}]$ whose $\bC$-linearization is the monoidal subcategory of $\cat{Rep}_\bC(U_q\mathfrak{sl}_2)$ generated by the defining module.  The Heisenberg-picture field theory $\int_\Box \cat{TL}$ packages together all of Kauffman-bracket skein theory.
\end{example}

\begin{remark}
  Using the same skein-theoretic technology, one can show:
  \begin{itemize}
  \item Every $\bK$-linear braided monoidal locally presentable category in which the unit object is compact projective and every object is a colimit of dualizable objects defines a framed Heisenberg-picture field theory $\cat{Bord}_3^{\mathrm{fr}} \to \cat{Alg}_2^{\mathrm{lax}}(\cat{Cocomp}_\bK)$.
  \item Every $\bK$-linear monoidal locally presentable category  in which the unit object is compact projective and every object is a colimit of dualizable objects defines a framed Heisenberg-picture field theory $\cat{Bord}_2^{\mathrm{fr}} \to \cat{Alg}_2^{\mathrm{lax}}(\cat{Cocomp}_\bK)$.
  \item Every $\bK$-linear monoidal locally presentable category  in which the unit object is compact projective and every object is a colimit of dualizable objects, which is additionally equipped with a ``pivotal'' structure, defines an oriented Heisenberg-picture field theory $\cat{Bord}_2^{\mathrm{or}} \to \cat{Alg}_2^{\mathrm{lax}}(\cat{Cocomp}_\bK)$.
  \end{itemize}
  After a preprint of this paper originally circulated, \cite{BJS18} proved versions of the above results using the Cobordism Hypothesis rather than skein theory.
\end{remark}


\end{document}